\documentclass[envcountsame,
envcountsect
]{llncs}

\usepackage{amsmath}
\usepackage{amssymb}
\usepackage{latexsym}
\usepackage{amsfonts} 
\usepackage[all]{xy}
\usepackage{stmaryrd}
\usepackage{wst14}
\usepackage
{hyperref}

\title{Proving Termination of Unfolding Graph Rewriting for General
Safe Recursion%
}

\author{Naohi Eguchi%
\thanks{The author is supported by Grants-in-Aid for JSPS Fellows (Grant No.
$25 \cdot 726$) that is granted at Graduate School of
Science, Chiba University, Japan.}%
}
\institute{Institute of Computer Science, University of Innsbruck%
\\
Technikerstrasse 21a, 6020 Innsbruck, Austria\\
\email{naohi.eguchi@uibk.ac.at}
}

\begin{document}

\maketitle

\begin{abstract}
In this paper we present a new termination proof and complexity analysis
of
{\em unfolding graph rewriting}
which is a specific kind of infinite graph rewriting expressing the general form of safe recursion.
We introduce a termination order over sequences of terms
together with an interpretation of term graphs into
sequences of terms.
Unfolding graph rewrite rules expressing general safe recursion can be
 successfully embedded into the termination order by the interpretation,
 yielding the polynomial runtime complexity.
Moreover, generalising the definition of unfolding graph
rewrite rules for general safe recursion, we propose a new criterion for
the polynomial runtime complexity of infinite GRSs and for the polynomial
size of normal forms in infinite GRSs.
\end{abstract}

\section{Introduction}

In this paper we present a new termination proof and complexity analysis
of a specific kind of infinite graph rewriting called 
{\em unfolding graph rewriting} \cite{GRR2010}.
The formulation of unfolding graph rewriting stems from a
function-algebraic characterisation of the polytime computable functions 
based on the principle known as {\em safe recursion} \cite{BC92} or
{\em tiered recursion} \cite{Leivant95}.
The schema of safe recursion is a syntactic restriction of the standard primitive
recursion based on a specific separation of argument positions of functions
into two kinds. 
Notationally, the separation is indicated by semicolon as 
$f(\sn{x_1, \dots, x_k}{x_{k+1}, \dots, x_{k+l}})$, 
where $x_1, \dots, x_k$ are called {\em normal} arguments while
$x_{k+1}, \dots, x_{k+l}$
are called {\em safe} ones.
The schema (\ref{e:sr}) formalises the idea that 
recursive calls is restricted on normal argument whereas
substitution of recursion terms is restricted for safe arguments:
\begin{equation}
\tag{\textbf{Safe Recursion}}
\label{e:sr}
\begin{array}{rcl}
f(0, \vec y; \vec z) &=& g(\vec y; \vec z) \\
f(c_i (x), \vec y; \vec z) &=& 
h_i (x, \vec y; \vec z, f(x, \vec y; \vec z))
\qquad (i \in I),
\end{array}
\end{equation}
where $I$ is a finite set of indices.
The purely function-algebraic characterisation in \cite{BC92} is
made more flexible and polynomial runtime complexity analysis is established in 
\cite{AM08,AEM12} in terms of termination orders.
As discussed in \cite{GRR2010}, safe recursion is sound for polynomial
runtime complexity over unary constructor, i.e., over numerals or sequences,
but it was not clear whether general forms of safe recursion over
arbitrary constructors, which is called 
{\em general ramified recurrence} \cite{GRR2010} or 
(\ref{e:gsr}), could be related to polytime computability.
\begin{equation}
\tag{\textbf{General Safe Recursion}}
\label{e:gsr}
f(c_i (x_1, \dots, x_{\arity (c_i)}), \vec y; \vec z) = 
h_i (\vec x, \vec y; \vec z, f(x_1, \vec y; \vec z), \dots,
                             f(x_{\arity (c_i)}, \vec y; \vec z)
    )
\ (i \in I)
\end{equation}
To see the difficulty of this question, consider a term rewrite system
(TRS for short) $\RS$ over the constructors 
$\{ \epsilon, \m{c}, \m{0}, \ms \}$
consisting of the following four rules with
the argument separation indicated in the rules.
\begin{equation*}
\label{ex:TRS}
\begin{array}{rclrcl}
\mg (\sn{\epsilon}{z}) & \rightarrow & z & \qquad
\mg (\sn{\m{c} (\sn{}{x, y})}{z}) & \rightarrow &
\m{c} (\sn{}{\mg (\sn{x}{z}), \mg (\sn{y}{z})}) \\
\mf (\sn{\m{0}, y}{}) & \rightarrow & \epsilon &
\mf (\sn{\ms (\sn{}{x}), y}) & \rightarrow & \mg (\sn{y}{\mf (\sn{x, y}{})})
\end{array}
\end{equation*}
%
Under the natural interpretation, 
$\mg (x, y)$ generates the binary tree appending the tree $y$ to every
leaf of the tree $x$, and 
$\mf (\ms^m (\m{0}), x)$ generates a tree consisting of exponentially
many copies of the tree $x$ measured by $m$.
Namely, rewriting in the TRS $\RS$ results in normal forms of exponential size
measured by the size of starting terms.
This problem cannot be solved by simple sharing. 
The authors of \cite{GRR2010} solved this problem,
showing that the equation of general safe recursion can be expressed by
an infinite set of unfolding graph rewriting.
As a consequence, the same authors answered the above question positively
in the sense as Theorem \ref{t:GRR10} in Section \ref{s:uf}.
In the present work, instead of looking at unfolding graph rewriting
sequences carefully, we propose complexity analysis by means of
termination orders over sequences of terms (Section \ref{s:order})
together with a successful embedding (Section \ref{s:pint}),
sharpening the complexity result obtained in \cite{GRR2010} 
(Section \ref{s:gsr}, Corollary \ref{c:gsr:2}).
In Section \ref{s:ptas} we generalise the definition of unfolding graph
rewrite rules for general safe recursion, we propose a new criterion for
the polynomial runtime complexity of infinite GRSs and for the polynomial
size of normal forms in infinite GRSs 
(Corollary \ref{c:spt:2}).

\section{Term graph rewriting}

In this section, we present basics of term graph rewriting following \cite{BareEGKPS87}.
Let $\FS$ be a {\em signature}, a finite set of function symbols, and
let 
$\arity : \FS \rightarrow \mathbb N$
where $\arity (f)$ is called the
{\em arity} of $f$.
We assume that $\FS$ be a signature partitioned into the set $\CS$ of
constructors and the set $\DS$ of defined symbols.
Let $G = (V_G, E_G)$ be a directed graph consisting of a set
$V_G$ of vertices (or nodes) and a set $E_G$ of directed edges.
A {\em labeled graph} is a triple 
$(G, \lab_G, \att_G)$
of an acyclic directed graph $G = (V_G, E_G)$,
a partial {\em labeling} function
$\lab_G: V_G \rightarrow \FS$ 
and a (total) {\em successor} function
$\att_G: V_G \rightarrow V_G^\ast$, 
mapping a node $v \in V_G$ to a sequence of nodes of length
$\arity (\lab_G)$,
such that
if
$\att_G (v) = v_1, \dots, v_k$, 
then
$\{ v_1, \dots, v_k \} = 
 \{ u \in V_G \mid (v, u) \in E_G \}$.
In case $\att_G (v) = v_1, \dots, v_k$,
the node $v_j$ is called the {\em $j$th successor of $v$} for every
$j \in \{ 1, \dots, k \}$.
A list 
$\nseq{v_1, m_1, \dots, v_{k-1}, m_{k-1}, v_k}$
consisting of nodes $v_1, \dots, v_m$ of a term graph $G$ and naturals
$m_1, \dots, m_{k-1}$ is called a {\em path} from $v_1$ to $v_k$
if $v_{j+1}$ is the $m_j$th successor of $v_j$
for each $j \in \{ 1, \dots, k-1 \}$.
A labeled graph $(G, \lab_G, \att_G)$ is {\em closed} if 
the labeling function $\lab_G$ is total.
Given two labeled graphs $G$ and $H$,
a {\em homomorphism} from $G$ to $h$ is a mapping
$\varphi: V_G \rightarrow V_H$ 
such that 
\begin{itemize}
\item $\lab_H (\varphi (v)) = \lab_G (v)$ for each 
      $v \in \dom{\lab_G} \subseteq V_G$, and
\item for each $v \in \dom{\lab_G}$, if $\att_G (v) = v_1, \dots, v_k$, then
      $\att_H (\varphi (v)) = \varphi (v_1), \dots, \varphi (v_k)$.
\end{itemize} 
By definition, these conditions are not required for a node $v \in V_G$ for which
$\lab_G (v)$ is not defined.
A quadruple $(G, \lab_G, \att_G, \rootnode_G)$ is a {\em term graph} if 
$(G, \lab_G, \att_G)$ is a labeled graph and
$\rootnode_G$ is a {\em root} of $G$, i.e.,
a unique node in $V_G$ from which every node is reachable.
We write $\mathcal{TG(F)}$ to denote the set of term graphs over a
 signature $\mathcal F$.
For a labeled graph 
$G = (G, \att_G, \lab_G)$
and a node $v \in V_G$,
$G \seg v$ denotes the sub-term graph of $G$ rooted at $v$.
A homomorphism $\varphi$ from a term graph $G$ to another term graph
$H$ is a homomorphism 
$\varphi : (G, \lab_G, \att_G) \rightarrow (H, \lab_H, \att_H)$
such that $\rootnode_H = \varphi (\rootnode_G)$.
Two term graphs $G$ and $H$ are {\em isomorphic}, denoted as
$G \cong H$, if there exists a bijective homomorphism from $G$ to $H$.
A {\em graph rewrite rule} is a triple $\rho = (G, l, r)$ of a labeled
graph $G$ and distinct two nodes $l$ and $r$ respectively called
the {\em left} and {\em right} root. 
The term rewrite rule
$\mg (x, y) \rightarrow \m{c} (y, y)$ 
is expressed by a graph rewrite rule (1) and
$\mh (x, y, z, w) \rightarrow \m{c} (z, w)$
is expressed by (2) in Figure \ref{fig:rule}.
In the examples, the left root is written in a circle while the right
root is in a square.
Undefined nodes are indicated as $\bot$.
Namely, undefined nodes behave as free variable.
\begin{figure}[t]
\begin{equation*}
  (1) \quad
  \xymatrix{*+[o][F-]{\mg} \ar[d] \ar[dr] &
            *+[F]{\m{c}} \ar@/_/[d] \ar@/^/[d] \\
             \bot & \bot
           }
  \qquad \qquad
  (2)
  \xymatrix{& *+[o][F-]{\mh} \ar[dl] \ar[d] \ar[dr] \ar[drr] & &
            *+[F]{\m{c}} \ar[dl] \ar[d] \\
            \bot & \bot & \bot & \bot
           }
\end{equation*}
\caption{Examples of graph rewrite rules}
\label{fig:rule}
\end{figure}
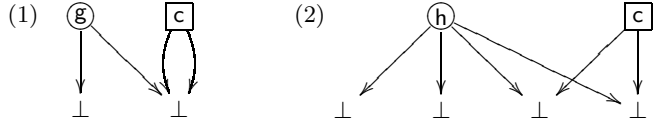 
A {\em redex} in a term graph $G$ is a pair 
$(R, \varphi)$ of a rewrite rule
$R = (H, l, r)$ and a homomorphism
$\varphi : H \seg l \rightarrow G$.
Intuitively, according to the homomorphism $\varphi$,
the subgraph $G \seg \varphi (l)$ to which $H \seg l$ is homomorphic by $\varphi$ is
replaced with the term graph to which $H \seg r$ is homomorphic.
A set $\GS$ of graph rewrite rules is called a 
{\em graph rewrite system} (GRS for short).
A graph rewrite rule $(G, l, r)$ is called a {\em constructor} one if
$\lab_G (l) \in \DS$ and
$\lab_G (v) \in \CS$ for any 
$v \in V_{G \seg l} \setminus \{ l \}$
whenever $\lab_G (v)$ is defined.
A GRS $\GS$ is called a constructor one if $\GS$ consists only of
constructor rewrite rules.
The rewrite relation defined by a GRS $\GS$ is denoted as
$\rightarrow_{\GS}$, its $m$-fold iteration as
$\rightarrow_{\GS}^m$, and its reflective and transitive closure is denoted as 
$\rightarrow_{\GS}^\ast$.
The {\em innermost} rewrite relation is defined in a
natural way, denoted as $\irew$, and $\irewm{m}$, $\irewast$ are defined
accordingly.

\section{Unfolding graph rewrite rules for general safe recursion}
\label{s:uf}

In this section we specify the shape of unfolding graph rewrite rules
which compatible with the schema of (\ref{e:gsr}).
We start with recalling the definition of unfolding graph rewrite rules
presented in \cite{GRR2010}.

\begin{definition}[Unfolding graph rewrite rules]
\label{d:uf}
\normalfont
Let $\Sigma$ and $\Theta$ be two disjoint signatures in bijective
 correspondence by 
$\varphi: \Sigma \rightarrow \Theta$.
For a fixed $k \in \mathbb N$, suppose that
$\arity ( \varphi (g)) = 2 \arity (g) + k$
for each $g \in \Sigma$.
Let
$f \not\in \Sigma \cup \Theta$ be a fresh function symbol
such that $\arity (f) = 1+k$.
Given a natural $m \geq 1$, an {\em unfolding graph rewrite rule over
$\Sigma$ and $\Theta$ defining $f$} is a graph rewrite rule
$\rho = (G, l, r)$ where
$G = (V_G, E_G, \att_G, \lab_G)$
is a labeled graph over a signature 
$\FS \supseteq \Sigma \cup \Theta$
that fulfills the following conditions.
\begin{enumerate}
\item The set $V_G$ of vertices consists of $1 + 2m + k$ elements
      $y$, $v_1, \dots, v_m$, $w_1, \dots$, $w_m$, $x_1, \dots, x_k$.
\item $l = y$ and $r = w_1$.
\item $\lab_G (y) = f$ and
      $\att_G (y) = v_1, x_1, \dots, x_k$.
\label{d:uf:f}
\item $\lab_G (x_j)$ is undefined for all
      $j \in \{ 1, \dots, k \}$.
\item For each $j \in \{ 1, \dots, m \}$,
      $\att_G (v_j) \subseteq \{ v_1, \dots, v_m \}^\ast$.
      Moreover,
      $V_{G \seg v_1} = \{ v_1, \dots, v_m \}$.
\label{d:uf:v}
\item For each $j \in \{ 1, \dots, m \}$,
      $\lab_G (v_j) \in \Sigma$ and
      $\lab_G (w_j) = \varphi (\lab_G (v_j))$.
\item For each $j \in \{ 1, \dots, m \}$,
      $\att_G (w_j) = 
      v_{j_1}, \dots, v_{j_n}, x_1, \dots, x_k, w_{j_1}, \dots, w_{j_n}$
      if
      $\att_G (v_j) = v_{j_1}, \dots, v_{j_n}$.
\label{d:uf:w}
\end{enumerate}
\end{definition}

\begin{example}
\label{ex:ugr}
Let $\Sigma = \{ \m{0}, \ms \}$,
$\Theta = \{ \mg, \mh \}$,
$\varphi : \Sigma \rightarrow \Theta$
be a bijection defined as
$\m{0} \mapsto g$ and $\ms \mapsto \mh$,
and
$\mf \not\in \Sigma \cup \Theta$, where
the arities of $\m{0}, \ms, \mg, \mh, \mf$
are respectively $0$, $1$, $1$, $3$ and $2$.
Namely we consider the case $k=1$.
The standard equations 
$\mf (\m{0}, x) \rightarrow \mg (x)$,
$\mf (\ms (y), x) \rightarrow \mh (y, x, \mf (y, x))$
for primitive recursion can be expressed by
the infinite set of unfolding graph rewrite rules over 
$\FS = \Sigma \cup \Theta \cup \{ \mf \}$ defining $f$, which includes
 the rewrite rules pictured in Figure 
\ref{fig:ugr}.
\begin{figure}[t]
\begin{equation*}
  \xymatrix{*+[o][F-]{\mf} \ar[d] \ar[dr] & & \\
            \m{0} & \bot & *+[F]{\mg} \ar[l]
           }
  \qquad \qquad
  \xymatrix{*+[o][F-]{\mf} \ar[d] \ar[ddr] & & \\
            \ms \ar[d] & & *+[F]{\mh} \ar@/_/[dll] \ar[dl] \ar[d] \\
            \m{0} & \bot & \mg \ar[l]
           }
  \qquad \qquad
  \xymatrix{*+[o][F-]{\mf} \ar[d] \ar[dddr] & & \\
            \ms \ar[d] & &
            *+[F]{\mh} \ar[dll] \ar[ddl] \ar[d] \\
            \ms \ar[d] & & \mh \ar[dll] \ar[dl] \ar[d] \\
            \m{0} & \bot & \mg \ar[l]
           }
\end{equation*}
\caption{Examples of unfolding graph rewrite rules}
\label{fig:ugr}
\end{figure}
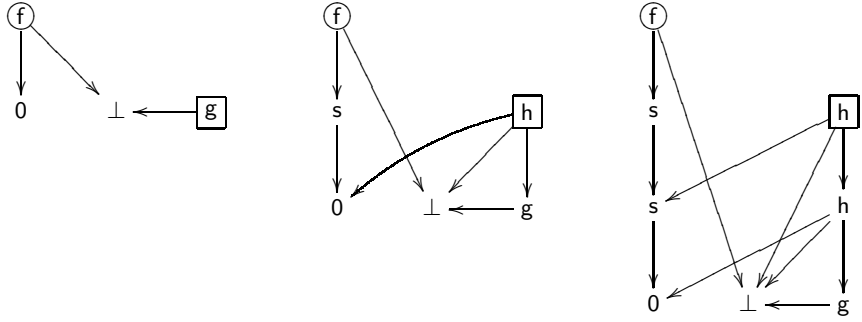
As seen from the pictures, the unfolding graph rewrite rules in Figure \ref{fig:ugr} express the infinite instances
$\mf (\m{0}, x) \rightarrow \mg (x)$,
$\mf (\ms (\m{0}), x) \rightarrow \mh (\m{0}, x, \mg (x))$,
$\mf (\ms (\ms (\m{0})), x) \rightarrow 
 \mh (\ms (\m{0}), x, \mh (\m{0}, x, \mg (x)))$, ...,
representing terms as term graphs.
\end{example}

In \cite{GRR2010} a graph rewrite system $\GS$ is called 
{\em polytime presentable} if there exists a deterministic polytime
algorithm which, given a term graph $G$, returns a term graph $H$ such
that $G \irew H$ if such a term graph exists, or
the value $\m{false}$ if otherwise.
In addition, a GRS $\GS$ is {\em polynomially bounded} if there exists a
polynomial $p: \mathbb N \rightarrow \mathbb N$ such that
$\max \{ m, |H| \} \leq p (|G|)$ holds whenever $G \irewm{m} H$ holds.
The main result in \cite{GRR2010} is restated as follows.

\begin{theorem}[Dal Lago, Martini and Zorzi \cite{GRR2010}]
\label{t:GRR10}
Every general safe recursive function can be represented by a
polytime presentable and polynomially bounded constructor GRS. 
\end{theorem}
In the proof of Theorem \ref{t:GRR10}, the case that a general safe
recursive function is
defined by (\ref{e:gsr}) is witnessed by an infinite set of unfolding
graph rewrite rules in a specific shape compatible with the 
argument separation as indicated in the schema
(\ref{e:gsr}).
To be compatible with the argument separation, in \cite{GRR2010}, for
any redex $(R, \varphi)$, the homomorphism $\varphi$ is restricted to
an {\em injective} one.
In this paper, instead of assuming injectivity of homomorphisms, we rather indicate the
argument separation explicitly.

\begin{definition}[Term graphs with the argument separation]
\label{d:safeterm}
\normalfont
In accordance with idea of safe recursion, we assume that the argument
 positions of every function symbol are separated into the normal
 and safe ones, writing 
$f(\sn{x_1, \dots, x_k}{x_{k+1}, \dots, x_{k+l}})$
to denote $k$ normal arguments and $l$ safe ones.
We always assume that every constructor symbol in $\CS$ has safe
argument positions only.
We take the argument separation into labeled graphs in such a way that
for every successor $u$ of a node $v$ 
we write $u \in \nrm (v)$ if $u$ is connected to a normal argument
 position of $\lab_G (v)$, and 
$u \in \safe (v)$ if otherwise.
For two distinct nodes $v_0$ and $v_1$, if
$\lab_G (v_0) = \lab_G (v_1)$, then,
for any $j \in \{ 1, \dots, \arity (\lab_G (v_0)) \}$,
$u_0 \in \nrm (v_0) \Leftrightarrow u_1 \in \nrm (v_1)$
for the $j$th successor $u_i$ of $v_i$ ($i = 0, 1$).
Notationally, we write 
$\att_G (v) =
 \sn{v_1, \dots, v_k}{v_{k+1}, \dots, v_{k+l}}$
to express the separation that
$v_1, \dots, v_k \in \nrm (v)$
and
$v_{k+1}, \dots, v_{k+l} \in \safe (v)$.
We assume that for any term graph 
$(G, \lab_G, \att_G, \rootnode_G)$ and for any node $v \in V_G$, either \ref{d:safeterm:1} or
\ref{d:safeterm:2} below holds.
\begin{enumerate}
\item For any path 
      $\nseq{v_0, m_0, \dots, v_{k-1}, m_{k-1}, v_k}$
      $(1 \leq k)$ in $G$, if $\lab_G (v_0) \in \DS$,
      $\lab_G (v_j) \in \CS$ for each 
      $j \in \{ 1, \dots, k-1 \}$, and
      $v_k = v$, then
      $v_1 \in \nrm (v_0)$.
\label{d:safeterm:1}
\item For any path 
      $\nseq{v_0, m_0, \dots, v_{k-1}, m_{k-1}, v_k}$
      $(1 \leq k)$ in $G$, if $\lab_G (v_0) \in \DS$,
      $\lab_G (v_j) \in \CS$ for each 
      $j \in \{ 1, \dots, k-1 \}$, and
      $v_k = v$, then
      $v_1 \in \safe (v_0)$.
\label{d:safeterm:2}
\end{enumerate} 
Accordingly, we assume that any homomorphism
$\varphi : G \rightarrow H$
preserves the argument separation.
Namely, for each $v \in \dom{\lab_G}$, if 
$\att_G (v) = v_1, \dots, v_k$; 
$v_{k+1}, \dots, v_{k+l}$, then
$\att_H (\varphi (v)) = 
 \sn{\varphi (v_1), \dots, \varphi (v_k)}%
    {\varphi (v_{k+1}), \dots, \varphi (v_{k+l})}$.
\end{definition}

\begin{example}
\label{ex:safeterm}
Let us consider graph representations of the term
$\mf ( \mg (\ms (\m{0}), \ms (\m{0}))$, 
$\mh (\ms (\m{0}), \ms (\m{0})))$.
All the graphs described in Figure \ref{fig:safeterm} are valid graph
 representations of the term.
On the other hand, consider the argument separation 
$\mf (\sn{x}{y})$, $\mg (\sn{x}{y})$, $\mh (\sn{x}{y})$
and the partition
$\CS = \{ \m{0}, \ms \}$, $\DS = \{ \mf, \mg, \mh \}$
of the signature.
Then, among the four graphs in Figure \ref{fig:safeterm}, the first
and second ones are valid representations of the term
$\mf (\sn{\mg (\sn{\ms (\m{0})}{\ms (\m{0})})}%
         {\mh (\sn{\ms (\m{0})}{\ms (\m{0})})}
     )$
but the others are not valid representations.
\begin{figure}[t]
\begin{equation*}
  \xymatrix{\mf \ar[d] \ar[drr] & & & &
            \mf \ar[d] \ar[dr] & & 
            \mf \ar[d] \ar[dr] & & 
            \mf \ar[d] \ar[dr] &
            \\
            \mg \ar[d] \ar[dr] & & \mh \ar[d] \ar[dr] & &
            \mg \ar[d] \ar[dr] & \mh \ar[dl] \ar[d] &
            \mg \ar@/_/[d] \ar@/^/[d] & \mh \ar@/_/[d] \ar@/^/[d] &
            \mg \ar@/_/[d] \ar@/^/[d] & \mh \ar@/_/[dl] \ar@/^/[dl]
            \\
            \ms \ar[d] & \ms \ar[d] & \ms \ar[d] & \ms \ar[d] &
            \ms \ar[d] & \ms \ar[d] &
            \ms \ar[d] & \ms \ar[d] &
            \ms \ar[d] &
            \\
            \m{0} & \m{0} & \m{0} & \m{0} &
            \m{0} & \m{0} & 
            \m{0} & \m{0} & 
            \m{0} &
           }
\end{equation*}
\caption{Examples of terms graphs with the argument separation}
\label{fig:safeterm}
\end{figure}
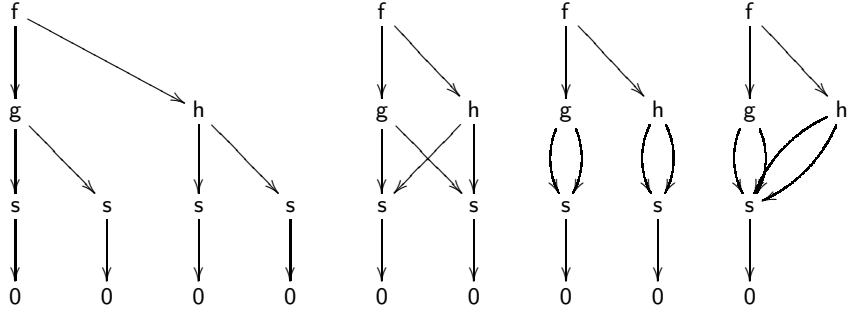
\end{example}

Let us recall the idea of safe recursion that the number of recursive
calls is measured only by a normal argument and recursion terms can be
substituted only for safe arguments.
This motivates us to introduce the following
safe version of unfolding graph rewrite rules.

\begin{definition}[Safe recursive unfolding graph rewrite rules]
\label{d:sruf}
\normalfont
We call an unfolding graph rewrite rule {\em safe recursive} if the
 following constraints imposed on the clause \ref{d:uf:f} and
 \ref{d:uf:w} in Definition \ref{d:uf} are satisfied.
\begin{enumerate}
\item In the clause \ref{d:uf:f},
      $v_1 \in \nrm (y)$.
\item In the clause \ref{d:uf:w},
      $v_{j_1}, \dots, v_{j_n} \in \nrm (w_j)$ and
      $w_{j_1}, \dots, w_{j_n} \in \safe (w_j)$.
\item In the clause \ref{d:uf:f} and \ref{d:uf:w},
      for each $j \in \{ 1, \dots, k \}$,
      $x_j \in \nrm (y)$ if and only if
      $x_j \in \nrm (w_i)$ for all $i \in \{ 1, \dots, m \}$.
\end{enumerate}
\end{definition}

As a consequence of Definition \ref{d:sruf}, we have a basic property of
safe recursive unfolding graph rewrite rules, which ensures that
rewriting by the graph rewrite rules does not change the structures of
subgraphs in normal argument positions.

\begin{corollary}
\label{c:sruf}
Let $(G, y, w_1)$ be a safe recursive unfolding graph rewrite rule with
the set $V_G$ of vertices consisting of $1+2m+k+l$ elements 
$y$, $v_1, \dots, v_m$, $w_1, \dots, w_m$, $x_1, \dots, x_{k+l}$
specified as in Definition \ref{d:uf} and \ref{d:sruf},
where 
$\att_G (y)$ $= \sn{v_1, x_1, \dots, x_k}{x_{k+1}, \dots, x_{k+l}}$.
Then, for any $j \in \{ 1, \dots, m \}$ and any node
$u \in \nrm (w_j)$, one of the following two cases holds.
\begin{enumerate}
\item If $\lab_G (u)$ is defined, then $u = v_i$ for some
      $i \in \{ 1, \dots, m \}$, and hence there exists a path
      from $v_1$ to $u$ in $G \seg y$.
\item If $\lab_G (u)$ is undefined, then $u = x_i$ for some
      $i \in \{ 1, \dots, k \}$, and hence $u \in \nrm (y)$.
\end{enumerate}
\end{corollary}

\section{Termination orders on sequences of terms}
\label{s:order}

In this section we consider a termination order 
$\spl{\ell}$ indexed by a positive natural $\ell$ over sequences of terms
based on an observation that every instance of unfolding graph rewrite
rules is precedence terminating in the sense defined in 
\cite{MidOZ96}.
Essentially, the termination order $\spl{\ell}$ is the same as 
{\em small polynomial path orders on sequences} \cite{AEM12}
but without recursive comparison.
We show that, for any fixed $\ell$, the length of any $\spl{\ell}$-reduction
sequence can be linearly bounded measured by the size of a starting term
but polynomially bounded if measured by $\ell$
(Lemma \ref{l:Gleq}). 

Let $\FS = \CS \cup \DS$ be a signature.
The set of terms over $\FS$ (and the set $\VS$ of variables) is denoted
as $\mathcal{T(F,V)}$, and the set of closed terms is denoted as
$\mathcal{T(F)}$.
We write $s \rhd t$ to express that $s$ is a {\em proper super-term} of $t$.
A {\em precedence} $\sp$ is a well founded partial binary relation on
$\FS$.
The {\em rank} $\rk : \FS \rightarrow \mathbb N$
is defined to be compatible with $>$:
$\rk (f) > \rk (g) \Leftrightarrow f > g$.
We always assume that every constructor symbol is $\sp$-minimal. 
To form sequences of terms, assume an auxiliary function symbol
$\circ$ whose arity is finite but arbitrary.
A term of the form 
$\circ (t_1, \dots, t_k)$
will be called a sequence if 
$t_1, \dots, t_k \in \mathcal{T(F,V)}$,
denoted as $\lseq[k]{t}$.
We will write $a, b, c, \dots$ for both terms and sequences.
We also write $\lseq[k]{s} \con \lseq[l]{t}$ to denote the concatenation
$\lst{s_1 \ \cdots \ s_k \ t_1 \ \cdots \ t_l}$.

\begin{definition}
\label{d:spl}
\normalfont
Let $\sp$ be a precedence on a signature $\FS$.
Suppose that $\ell \in \mathbb N$ and $1 \leq \ell$.
Then $a \spl{\ell} b$ holds if one of the following three cases holds.
\begin{enumerate}
\item $a = f(s_1, \dots, s_k)$, $b = g(t_1, \dots, t_l)$,
       $f, g \in \FS$, $f > g$,
  \begin{itemize}
  \item $f(s_1, \dots, s_k) \rhd t_j$
        for all $j \in \{ 1, \dots, k \}$, and
  \item $l \leq \ell$.
  \end{itemize}
\label{d:spl:st}
\item $a = f(s_1, \dots, s_k)$, $f \in \FS$, $b = \lst{t_1 \cdots t_l}$,
  \begin{itemize}
  \item $f(s_1, \dots, s_k) \spl{\ell} t_j$
        for all $j \in \{ 1, \dots, l \}$, and
  \item $l \leq \ell$.
  \end{itemize}
\label{d:spl:sb}
\item $a = \lst{s_1 \cdots s_k}$, $b = \lst{t_1 \cdots t_l}$
      and there exists a permutation
      $\pi : \{ 1, \dots, l \} \rightarrow \{ 1, \dots, l \}$,
      and there exist terms or sequences
      $b_j$ $(j = 1, \dots, k)$ such that
  \begin{itemize}
  \item $b_1 \con \cdots \con b_k = \lst{t_{\pi (1)} \cdots t_{\pi (l)}}$,
  \item $s_j \eqspl{\ell} b_j$ for all $j \in \{ 1, \dots, k \}$, and
  \item $s_i \spl{\ell} b_i$ for some $i \in \{ 1, \dots, k \}$.
  \end{itemize} 
In case some $b_i$ is a term $t$, the concatenation
$\cdots \con b_i \con \cdots$ 
should be understood as
$\cdots \con \lst{t} \con \cdots$.
\label{d:spl:ab}
\end{enumerate}
\end{definition}
For notational convention, we write $a \caseref{\spl{\ell}}{i} b$ if 
$a \spl{\ell} b$ follows from the $i$-th clause in Definition
\ref{d:spl}. 
Note for example that 
if $s \cspl{sb} \lseq[l]{t}$, then
$s \cspl{st} t_j$ holds for all $j \in \{ 1, \dots, l \}$.
As a special case,
$s \cspl{sb} \lst{}$ holds
for any $\ell \geq 2$ and any term $s$, and hence
$a \cspl{ab} \lst{}$ holds
for any non-empty sequence $a$.

\begin{lemma}
\label{l:Gl}
\begin{enumerate}
\item If $a \spl{\ell} b$ and $\ell \leq \ell'$, then 
      $a \spl{\ell'} b$ holds.
\label{l:Gl:mon}
\item If $a \spl{\ell} a'$ holds, then
      $b \con a \con c \spl{\ell} b \con a' \con c$
      also holds.
\label{l:Gl:con}
\item If $s \cspl{sb} b$, $b = b_0 \con b_1$ and
      $b_j \neq \lst{}$ for each $j = 0, 1$, then
      $\lst{s} \con c \cspl{ab} b_0 \con c \con b_1$
      holds for any sequence $c$.
\label{l:Gl:perm}
\end{enumerate}
\end{lemma}

\begin{proof}
Property \ref{l:Gl:mon} and \ref{l:Gl:con} can be shown as in 
\cite[Lemma 9]{AEM12}.
Consider Property \ref{l:Gl:perm}.
The case $c = \lst{}$ immediately follows from the assumption.
Let $b_0 = \lst{t_1 \cdots t_{k_0}}$,
$b_1 = \lst{t_{k_0 +l+1} \cdots t_{k_0 +l+ k_1}}$, and
$c = \lst{t_{k_0 +1} \cdots t_{k_0 + l}}$.
Define a permutation on
$\pi : \{ 1, \dots, k_0 + l + k_1 \}$ 
by
\begin{equation*}
  \pi (j) =
  \begin{cases}
  j & \text{if } j \in \{ 1, \dots, k_0 \}, \\
  j + k_1 & \text{if } j \in \{ k_0 +1, \dots, k_0 + l \}, \\
  j - l & \text{if } j \in \{ k_0 +l+1, \dots, k_0 +l+ k_1 \}.
  \end{cases}
\end{equation*}
Then we have
$b_0 \con c \con b_1 = \lst{t_1 \cdots t_{k_0 +l+ k_1}}$
and
$b \con c =
 \lst{t_{\pi (1)} \cdots t_{\pi (k_0 +l+ k_1)}}$.
Now
$a \con c \cspl{ab} b_0 \con c \con b_1$
follows from $s \spl{\ell} b$ and
$t_j \eqspl{\ell} t_j$ for every
$j \in \{ k_0 +1, \dots, k_0 +l \}$.
\qed
\end{proof}

Since the relation $\spl{\ell}$ can be regarded as a fragment of small polynomial path orders
on sequences defined in \cite{AEM12},
$\spl{\ell}$ is well founded
for any fixed $\ell \geq 1$.
Therefore the following complexity measure 
$\mG{\ell}: \mathcal{T} \rightarrow \mathbb N$ 
can be well defined.

\begin{definition}
$\mG{\ell} (a) := \max 
 \{ k \in \mathbb N \mid \exists a_1, \dots, a_k \mbox{ such that }
    a \spl{\ell} a_1 \spl{\ell} \cdots \spl{\ell} a_k
 \}$
\end{definition}
Note that $\mG{\ell} (a) > \mG{\ell} (b)$ holds whenever $a \spl{\ell} b$ holds.
As employed in \cite{AEM12}, the following basic of $\mG{\ell}$
property can be shown, whose proof can be found in \cite[Lemma 7]{AEMspop}.

\begin{lemma}
\label{l:Gl:sum}
For any $\ell \geq 1$ and sequence $a = \lst{t_1 \cdots t_k}$,
      $\mG{\ell} (a) = \sum_{j=1}^{k} \mG{\ell} (t_j)$ holds.
\end{lemma}

\begin{lemma}
\label{l:Gleq}
Let $\ell \geq 1$ and
$\max \{ \arity (f) \mid f \in \FS \} \leq d$.
Then, for any function symbol $f \in \FS$ with arity $k \leq \ell$
and for any closed terms $s_1, \dots, s_k \in \mathcal{T(C)}$, the
 following inequality holds, where $\depth (t)$ denotes the {\em depth}
 of a term $t$ in the standard tree representation.
\begin{equation*}
\textstyle
\mG{\ell} (f(s_1, \dots, s_k)) \leq
d^{\rk (f)} \cdot (1+ \ell)^{\rk (f)} \cdot 
\left( 1 + \sum_{j=1}^k \depth (s_j) \right).
\end{equation*}
\end{lemma}

\begin{proof}
Let $s = f(s_1, \dots, s_k)$.
We show the lemma by induction on 
$\rk (f)$.
In the base case $\rk (f) = 0$, all the possible reduction is
$f(s_1, \dots, s_k) \spl{\ell} \lst{}$, and hence
$\mG{\ell} (s) \leq 1$.
For the induction step, suppose $\rk (f) > 0$.
It suffices to show that for any $b$, if 
$s \spl{\ell} b$, then
$\mG{\ell} (b) <
d^{\rk (f)} \cdot (1 + \ell)^{\rk (f)} \cdot
\left( 1 + \sum_{j=1}^k \depth (s_j) \right)
$
holds.
This is shown by case analysis splitting into
$s \cspl{st} b$ and $s \cspl{sb} b$

{\sc Case.} $s \cspl{st} b = g(t_1, \dots, t_l)$:
In this case, $f >_{\FS} g$, 
$s \rhd t_j$ for all $j \in \{ 1, \dots, l \}$,
and $l \leq \ell$.
Since $\rk (f) > \rk (g)$, 
the induction hypothesis yields
$\mG{\ell} (b) \leq 
d^{\rk (g)} \cdot (1 + \ell)^{\rk (g)} \cdot
\left( 1 + \sum_{j=1}^l \depth (t_j) \right)$.
On the other hand,
$1 + \sum_{j=1}^l \depth (t_j) \leq d
 \left( 1 + \sum_{j=1}^k \depth (s_j) \right)
$,
and hence
\begin{eqnarray}
\mG{\ell} (b) &\leq&
\textstyle
d^{\rk (g)} \cdot (1 +\ell)^{\rk (g)} \cdot
d \left( 1 + \sum_{j=1}^k \depth (s_j) \right)
\nonumber
\\
&\leq&
\textstyle
d^{\rk (f)} \cdot (1 + \ell)^{\rk (f)-1} \cdot
\left( 1 + \sum_{j=1}^k \depth (s_j) \right).
\label{eq:Gleq}
\end{eqnarray}

{\sc Case.} $s \cspl{sb} b = \lst{t_1 \cdots t_l}$:
In this case,
$l \leq \ell$ and 
$s \cspl{st} t_j$ for all $j \in \{ 1, \dots, l \}$.
By (\ref{eq:Gleq}) in the previous case,
$\mG{\ell} (t_j) \leq 
d^{\rk (f)} \cdot (1 + \ell)^{\rk (f)-1} \cdot
\left( 1 + \sum_{j=1}^k \depth (s_j) \right)
$
holds for all $j \in \{ 1, \dots, l \}$.
Therefore
\begin{eqnarray*}
\mG{\ell} (b) &\leq&
\textstyle
\ell \cdot
d^{\rk (f)} \cdot (1 + \ell)^{\rk (f)-1} \cdot
\left( 1 + \sum_{j=1}^k \depth (s_j) \right) 
\qquad (\text{by Lemma \ref{l:Gl}.\ref{l:Gl:sum}})
\\
&<&
\textstyle
d^{\rk (f)} \cdot (1+ \ell)^{\rk (f)} \cdot 
\left( 1 + \sum_{j=1}^k \depth (s_j) \right).
\end{eqnarray*}
\qed
\end{proof}

\section{Predicative embedding of safe recursive unfolding graph rewriting into 
$\spl{\ell}$}
\label{s:pint}

In this section we present the {\em predicative} interpretation of term
graphs into sequences of terms, showing that, by the interpretation,
rewriting sequences by safe recursive unfolding graph rewrite rules can
be embedded into the termination order
$\spl{\ell}$ presented in the previous section (Theorem \ref{t:context}). 
This yields that the length of any rewriting
sequence by safe recursive unfolding graph rewrite rules starting with a
term graphs whose arguments are already normalised can be bounded by a polynomial in the sizes of the normal
argument subgraphs only.
The predicative interpretation is defined modifying
the predicative interpretations for terms, which
stem from \cite{AM05} and are employed in \cite{AM08,AEM11,AEM12}.

\begin{definition}
\label{d:safepath}
\normalfont
\begin{enumerate}
\item A path $\nseq{v_1, m_1, \dots, v_{k-1}, m_{k-1}, v_k}$
      in a term graph $G$ is called a {\em safe} one if
      $v_{j+1} \in \safe (v_j)$ for all $j \in \{ 1, \dots, k-1 \}$.
\label{d:safepath:1}
\item Given a signature $\FS = \CS \cup \DS$,
      we define a subset 
      $\mathcal{TG}_{\nrm} (\FS) \subseteq \mathcal{TG(F)}$.
      Let $G \in \mathcal{TG(F)}$ with 
      $\att_G (\rootnode_G) = 
       \sn{v_1, \dots, v_k}{v_{k+1}, \dots, v_{k+l}}$.
      Then $G \in \mathcal{TG}_{\nrm} (\FS)$ if
      $G \in \mathcal{TG(C)}$, or
      $G \seg v_j \in \mathcal{TG(C)}$ for each
      $j \in \{ 1, \dots, k \}$ and
      $G \seg v_j \in \mathcal{TG}_{\nrm} (\FS)$ for each
      $j \in \{ k+1, \dots, k+l \}$.
\label{d:safepath:2}
\end{enumerate}
\end{definition}

\begin{lemma}     
\label{l:safepath}
Let $\GS$ be a set of safe recursive unfolding graph rewrite rules over
 a signature $\FS$ and
$G \in \mathcal{TG}_{\nrm} (\FS)$.
\begin{enumerate}
\item  For any redex $(R, \varphi)$ in $G$ with a rewrite rule
$R = (H, l, r) \in \GS$ and a homomorphism 
$\varphi : H \seg l \rightarrow G$,
the node $\varphi (l)$ lies on a safe path from $\rootnode_G$ in $G$.
\label{l:safepath:1}
\item If $G \rew H$, then
$H \in \mathcal{TG}_{\nrm} (\FS)$.
\label{l:safepath:2}
\end{enumerate} 
\end{lemma}

\begin{proof}
{\sc Property} \ref{l:safepath:1}.
Assume that $\varphi (l)$ is not on any safe path from $\rootnode_G$.
Then, there exists a path 
$\nseq{v_0, m_0, \dots, v_{k-1}, m_{k-1}, v_k}$
such that
$v_k = \varphi (l)$,
$\lab_G (v_j) \in \CS$ for each 
$j \in \{ 1, \dots, k-1 \}$
and
$v_1 \in \nrm (v_0)$.
Since constructor symbols only have safe arguments, it holds that
$\lab_G (v_0) \in \DS$.
Hence, by the condition \ref{d:safeterm:1} in Definition
 \ref{d:safeterm},
for any path
$\nseq{u_0, n_0, \dots, u_{l-1}, n_{l-1}, u_l}$
      $(1 \leq k)$ in $G$, 
if $u_l = \varphi (l)$, $\lab_G (u_0) \in \DS$, and
      $\lab_G (u_j) \in \CS$ for each 
      $j \in \{ 1, \dots, k-1 \}$, then
      $u_1 \in \nrm (u_0)$.
This means $G \seg v_1 \in \mathcal{TG(C)}$ by the definition of the
 class
$\mathcal{TG}_{\nrm} (F)$, and thus
$G \seg \varphi (l) = (G \seg v_1) \seg \varphi (l)$ is not rewritable. 

{\sc Property} \ref{l:safepath:2}.
Suppose that $G$ results in $H$ by applying a redex 
$(R, \varphi)$ for a rule $R = (G_0, l, r) \in \GS$.
Since $\varphi (l)$ lies on a safe path from $\rootnode_G$ by Property
 \ref{l:safepath:1}, it suffices to show that 
$H \seg r' \in \TGnrm$ for the node 
$r' \in H$ corresponding to $r \in V_{G_0}$.
Let $v \in V_{H \seg r'} \setminus \{ r' \}$.
By structural induction over $H \seg v$,
employing Corollary \ref{c:sruf},
it can be shown that
$H \seg v \in \mathcal{TG(C)}$ if $v \in \nrm (u)$ for some
$u \in V_{H \seg r'}$,
and
$H \seg v \in \TGnrm$ if $v \in \safe (u)$ for some
$u \in V_{H \seg r'}$.
Then
$H \seg r' \in \TGnrm$
follows accordingly.
\qed
\end{proof}

\begin{definition}[Interpretation of term graphs into unlabeled graphs]
\label{d:pj}
\normalfont
In order to define the predicative interpretation, we define an
 interpretation $\pj$ of term graphs into {\em unlabeled graphs}.
For a term graph $G$, $\pj (G)$ denotes the directed graph 
$(V_{\pj (G)}, E_{\pj (G)})$ with the root
$\rootnode_{\pj (G)} = \rootnode_G$ consisting of the set 
$V_{\pj (G)} =  V_G$ of vertices, and the set 
$E_{\pj (G)}$ of edges defined as follows.
For an edge $(u, v) \in E_G$, $(u, v) \in E_{\pj (G)}$ holds if either 
\ref{d:pj:1} or \ref{d:pj:2} below holds.
\begin{enumerate}
\item There are no distinct two safe paths
      from $\rootnode_G$ to $v$.
\label{d:pj:1}
\item The edge $(u, v)$ lies on a safe path 
      $\nseq{u_1, m_1, \dots, u_{k-1}, m_{k-1}, v}$
      from $\rootnode_G$ to $v$, i.e., $u_1 = \rootnode_G$ and $u_{k-1} = u$,
      and, for any distinct safe path  
      $\nseq{v_1, n_1, \dots$, $v_{l-1}, n_{l-1}, v}$
      from $\rootnode_G$ to $v$,
      $m_i < n_j$ holds
      whenever $u_i = v_j$ and $m_i \neq n_j$.
      Namely, a safe path is kept by the interpretation $\pj$ if it is the
      leftmost one. 
\label{d:pj:2}
\end{enumerate}
\end{definition}

\begin{example}
\label{ex:pj}
Let us consider a term graph $G$ pictured in Figure \ref{fig:pj} with 
$\att_G (v_0)$ $= \sn{u_1, u_2}{v_1, v_2}$
and
$\att_G (v_1) = \sn{u_2}{v_2, v_2}$.
Since $u_1, u_2 \in \nrm (v_0)$ and $u_2 \in \nrm (v_1)$,
all of the edges
$(v_0, u_1)$, $(u_1, u_2)$, $(v_0, u_2)$ and $(v_1, u_2)$ are trivially preserved.
The path $\nseq{v_0, 3, v_1}$ is the unique safe path from $\rootnode_G$ to 
$v_1$, and hence the edge $(v_0, v_1)$ is preserved as well. 
Consider the edge $(v_1, v_2) \in E_G$.
There are two distinct safe paths 
$\nseq{v_0, 3, v_1, 2, v_2}$ and 
$\nseq{v_0, 3, v_1, 3, v_2}$
from $\rootnode_G$ to $v_2$.
The node $v_1$ lies on the both safe paths, and hence 
the edge $(v_1, v_2)$ is also preserved.
Finally, consider the edge
$(v_0, v_2) \in E_G$.
There are three distinct safe paths
$\nseq{v_0, 3, v_1, 2, v_2}$,
$\nseq{v_0, 3, v_1, 3, v_2}$
and
$\nseq{v_0, 4, v_2}$
from $\rootnode_G$ to $v_2$.
The edge $(v_0, v_2)$ lies only on the last one, which is not the
 leftmost, and thus the edge $(v_0, v_2)$ is not
 preserved.
Summing up, we obtain the unlabeled graph $\pj (G)$ as pictured in Figure \ref{fig:pj}.
\begin{figure}[t]
\begin{equation*}
  \xymatrix{ & & v_0 \ar[dl] \ar[ddl] \ar@/_/[d] \ar@/^3mm/[dd] & \\
            G = & u_1 \ar[d] & v_1 \ar[dl] \ar@/_/[d] \ar[d] &
            \mapsto \\
            & u_1 & v_2 &
           }
\qquad
  \xymatrix{ & & v_0 \ar[dl] \ar[ddl] \ar[d] \\
            \pj (G) = & u_1 \ar[d] & v_1 \ar[dl] \ar[d] \\
            & u_1 & v_2
           }
\end{equation*}
\caption{Interpretation of term graphs into unlabeled graphs}
\label{fig:pj}
\end{figure}
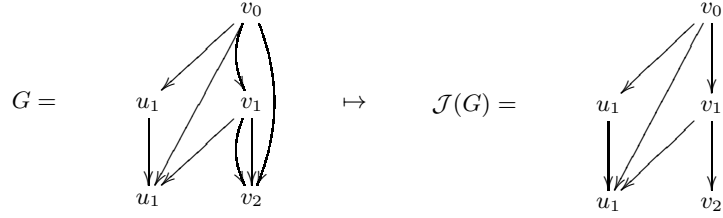
\end{example}

For each function symbol $f \in \FS$ with $k$ normal argument
positions, let $\fn$ denote a fresh function symbol with $k$ argument
positions.
We write $\FSn$ to denote the new signature
$\{ \fn \mid f \in \FS \}$.
For a term graph $G$, we write 
$\term (G)$ to denote the standard term representation of $G$, i.e.,
$\term (G) = \lab_G (\rootnode_G)
 (\sn{\term (G \seg v_1), \dots, \term (G \seg v_k)}%
     {\term (G \seg v_{k+1}), \dots, \term (G \seg v_{k+l})})$
if
$\att_G (\rootnode_G) =
 \sn{v_1, \dots, v_k}{v_{k+1}, \dots, v_{k+l}}$.
For two successors $v_0, v_1$ of a node $v$, if $v_j$ is the $k_j$th
successor for each $j \in \{ 0, 1 \}$ and $k_0 < k_1$, then we write
$v_0 <_{(G, v)} v_1$, or simply write $v_0 < v_1$ if no confusion likely arises.
Furthermore, we extend the notation $G \seg v$ to unlabeled (acyclic)
directed graphs in the most natural way.

\begin{definition}[Predicative interpretation of term graphs]
\label{d:pint}
\normalfont
Let $G$ be a closed term graph over a signature $\FS = \CS \cup \DS$,
$f = \lab_G (\rootnode_G)$, and
$\att_G (\rootnode_G)$ $=$ 
$\sn{v_1, \dots, v_k}{v_{k+1}, \dots, v_{\arity (f)}}
$.
Suppose that
$\{ u_1, \dots, u_n \} = 
 \{ v \in V_{G} \mid v \in \safe (\rootnode_G)
$ and  
$  (\rootnode_{G}, v) \in E_{\pj (G)}
 \}$
and $u_1 < \cdots < u_n$.
Then we define an interpretation 
$\PINT : \mathcal{TG(F)} \rightarrow 
 \mathcal{T(\FS \cup \FSn \cup \{ \circ\})}$
by
\begin{equation*}
\PINT (G) = 
  \begin{cases}
  \lst{} \quad (\text{the empty sequence})
  \quad \hfill \text{if } G \in \mathcal{TG(C)},
  & \\
  \lst{\fn (\term (G \seg v_1), \dots, \term (G \seg v_k))} \con
  \PINT (G \seg u_1) \con \cdots \con \PINT (G \seg u_n)
  \text{ o.w.} &
  \end{cases}
\end{equation*}
\end{definition}
We note that, for a node $v \in G$, $\pj (G \seg v) = \pj (G) \seg v$
does not hold in general.
Thus it should be understood that the result $\pint (G \seg v)$ of the
interpretation depends on $\pj (G) \seg v$ not on $\pj (G \seg v)$.

\begin{example}
\label{ex:pint}
Consider again the term graph $G$ in Example \ref{ex:pj}.
Let $\mf = \lab_G (v_0)$, $\ms = \lab_G (u_1)$,
$\m{0} = \lab_G (u_2)$, $\mh = \lab_G (v_1)$,
$\mg = \lab_G (v_2)$, and $\mh, \mg \in \DS$.
Then inductively one can see that the interpretation $\pint$ works for 
$G$ as follows.

$
\begin{array}{rcl}
\pint (G \seg v_2) &=& \lst{\fsn{\mg}}, \\
\pint (G \seg v_1) &=& \lst{\fsn{\mh} (\m{0})} \con \pint (G \seg v_2)
= \lst{\fsn{\mh} (\m{0}) \ \fsn{\mg}}, \\
\pint (G) &=& \lst{\fsn{\mf} (\ms (\m{0}), \m{0})} \con \pint (G \seg v_1)
= \lst{\fsn{\mf} (\ms (\m{0}), \m{0}) \ \fsn{\mh} (\m{0}) \ \fsn{\mg}}.
\end{array}
$

\noindent
If the interpretation $\pj$ is not performed, then
$G$ would be translated into the sequence
$\lst{\fsn{\mf} (\ms (\m{0}), \m{0}) \ \fsn{\mh} (\m{0}) \ \fsn{\mg} \ 
       \fsn{\mg} \  \fsn{\mg}}$, 
in which the term $ \fsn{\mg}$ is duplicated unnecessarily.
\end{example}


\begin{lemma}
\label{l:pint}
Let $\GS$ be a set of safe recursive constructor unfolding graph rewrite
 rules over a signature $\FS$.
Suppose that $G \rew H$ is induced by a redex $(R, \varphi)$ in a closed term graph 
$G \in \mathcal{TG}_{\nrm} (\FS)$
for a rule $R = (G', l, r) \in \GS$ and a homomorphism
$\varphi : G' \seg l \rightarrow G$.
Let $r' \in V_H$ the node corresponding to $r \in V_{G'}$.
Then, for the interpretations $\pint$ defined for $G$ and $H$,
$\PINT (G \seg \varphi (l)) \spl{\ell} 
 \PINT (H \seg r')$
holds for
$\ell =
\max (\{ |G' \seg r| \} \cup \{ \arity (f) \mid f \in \FS \})$.
\end{lemma}

\begin{proof}
Let $R = (G', l, r) \in \GS$ define a
function symbol $f$ over $\FS \supseteq \Sigma \cup \Theta$
with $\arity (f) = 1 + k + l$.
In case $H \seg r' \in \mathcal{TG(C)}$,
clearly
$\pint (G \seg \varphi (l)) \cspl{ab} \lst{} =
 \pint (H \seg r')$
holds since $G \seg \varphi (l) \not\in \mathcal{TG(C)}$.
In the sequel, we suppose $H \seg r' \not\in \mathcal{TG(C)}$.
Let the set $V_{G'}$ of vertices consist of
$y$, $v_1, \dots$, $v_m, w_1, \dots, w_m$, $x_1, \dots, x_k$, 
$x_{k+1}, \dots, x_{k+l}$ 
as specified in Definition \ref{d:uf}
and \ref{d:sruf}, where
$\{ x_1, \dots, x_k \} \subseteq \nrm (y)$ and
$\{ x_{k+1}, \dots, x_{k+l} \} \subseteq \safe (y)$
hold.
In particular, $l = y$, $r = w_1$ and $\lab_{G'} (l) = f$ hold by definition.
To make the presentation simpler, let us identify the nodes
$y$, $v_1, \dots, v_m$, $x_1, \dots, x_k$,
$x_{k+1}, \dots, x_{k+l} \in V_{G' \seg l}$
with the nodes in $V_{G}$ corresponding by the homomorphism
$\varphi$
and the nodes
$v_2, \dots, v_m$, $w_1, \dots, w_m$, $x_1, \dots, x_k$,
$x_{k+1}, \dots, x_{k+l} \in V_{G' \seg r}$
with the corresponding nodes in $V_H$, 
e.g., 
$y = \varphi (l)$ and $w_1 = r'$.
We write $g$ to denote 
$\lab_{H} (w_1)$.
Then, by the interpretations $\pint$ defined for $G$ and $H$, the term graphs $G \seg y$ and $H \seg w_1$ are respectively transformed into
 the following sequences of terms.
\begin{eqnarray*}
\PINT (G \seg y) &=&
\lst{\fn (\term (G \seg v_1), 
          \term (G \seg x_1), \dots, \term (G \seg x_k))
    }
\\
&& \hspace{5.5cm}
\con 
\PINT (G \seg z_{1}) \con \cdots \con
\PINT (G \seg z_{l'}),
\\
\PINT (H \seg w_1) &=&
[\gn (\term (H \seg v_{j_1}), \dots, \term (H \seg v_{j_n}),
          \term (H \seg x_1), \dots, \term (H \seg x_k))]
\\
&& \hspace{1cm}
\con 
\PINT (H \seg z_{1}) \con \cdots \con
\PINT (H \seg z_{l'})
\con 
\PINT (H \seg u_{1}) \con \cdots \con
\PINT (H \seg u_{n'}),
\end{eqnarray*}
where 
$\att_G (v_1) = v_{j_1}, \dots, v_{j_n}$,
and
$z_1, \dots, z_{l'}$
and
$u_1, \dots, u_{n'}$
denotes the sequence of nodes such that 
\begin{itemize}
\item $\{ z_1, \dots, z_{l'} \} =
 \{ v \in \{ x_{k+1}, \dots, x_{k+l} \} \mid
   (y, v) \in E_{\pj (G) \seg y}
 \}
=
 \{ v \in \{ x_{k+1}, \dots$, $x_{k+l} \} \mid
   (w_1, v) \in E_{\pj (H) \seg w_1}
 \}
$,
$z_1 < \cdots < z_{l'}$, and
\item $\{ u_1, \dots, u_{n'} \} =
 \{ v \in \{ w_{j_1}, \dots, w_{j_n} \} \mid
   (w_1, v) \in E_{\pj (H) \seg w_1}
 \}
$
and $u_1 < \cdots < u_{n'}$.
\end{itemize}
Define a precedence $>$ over $\FSn$ as
$\fn > \fsn{h}$ for any $h \in \Theta$.
Write $s$ to denote $\term (G' \seg v_1)$, 
$s_j$ to denote
$\term (G \seg x_j)$ for each $j \in \{ 1, \dots, k \}$,
$t_i$ to denote
$\term (H \seg v_{j_i})$ for each $i \in \{ 1, \dots, n' \}$,
and write
$t_j'$ to denote
$\term (H \seg x_{j})$ for each $i \in \{ 1, \dots, k \}$.
First we show that
$\fn (s, s_1, \dots, s_k) \spl{\ell} 
 \gn (t_{1}, \dots, t_{n}, t'_1, \dots, t'_k)$
holds.
Since $V_{G \seg v_1} = \{ v_1, \dots, v_m \}$ by definition,
any of $H \seg v_{j_1}, \dots, H \seg v_{j_n}$ is a subgraph of
$G \seg v_1$, and hence
any of $t_{1}, \dots, t_{n}$ is a subterm of
$s$.
This yields
$\fn (s, s_1, \dots, s_k) \rhd t_{i}$
for all $i \in \{ 1, \dots, n\}$.
Moreover, since
$t_i = t'_i$ for every $i \in \{ 1, \dots, k \}$,
$\fn (s, s_1, \dots, s_k) \rhd t'_{i}$
also holds for all $i \in \{ 1, \dots, k \}$.
These together with $\fn > \gn$ and 
$\arity (\gn) \leq \arity (\mg)$ imply
\begin{equation}
 \fn (s, s_1, \dots, s_k)
  \caseref{\spl{d}}{\ref{d:spl:st}}
 \gn (t_{1}, \dots, t_{n}, t'_1, \dots, t'_k),
\label{e:pint:1}
\end{equation} 
where 
$d := \max \{ \arity (f) \mid f \in \FS \}$.  
For $i \in \{ 1, \dots, m \}$, let
$\ell_i = \max \{ |V_{\safe ( w_i)}|, d \}$,
where
\[
 V_{\safe (u)} := 
 V_{\pj (H) \seg u} \cap 
 \{ w_1, \dots, w_m
 \}.
\]
By structural induction over $H \seg w_i$, one can show that
\begin{equation}
 \fn (s, s_1, \dots, s_k) \caseref{\spl{\ell_i}}{\ref{d:spl:sb}} 
 \PINT (H \seg w_i)
\label{e:pint:2}
\end{equation}
holds for every $i \in \{ 2, \dots, m \}$
for the interpretation $\pint$ defined for $H$.
The base case can be shown in the same way as we proved 
(\ref{e:pint:1}).
Since
$\att_H (w_1) =
 \sn{v_{j_1}, \dots, v_{j_n}, x_1, \dots, x_k}%
    {x_{k+1}, \dots, x_{k+l}, w_{j_1}, \dots, w_{j_n}}$,
there is no edge
$(w_i, x_{j}) \in E_{\pj (H)}$
for any 
$i \in  \{ 2, \dots, m \}$
and
$j \in  \{ 1, \dots, k \}$.
Hence, instead of proving the general induction step, it would suffice to show the
 following orientation assuming (\ref{e:pint:2}).
\begin{equation}
\fn (s, s_1, \dots, s_k)
\caseref{\spl{\ell_1}}{\ref{d:spl:sb}}
\lst{\gn (t_{1}, \dots, t_{n}, t'_1, \dots, t'_k)}
\con \PINT (G' \seg u_{1}) \con \cdots \con
\PINT (G' \seg u_{n'}).
\label{e:pint:3}
\end{equation}
By the definition of the interpretation $\pj$,
$V_{\safe (u_1)}, \dots, V_{\safe (u_{n'})}$
are pair-wise disjoint, and hence
$|V_{\safe (w_1)}|  = 
 1 + \sum_{j=1}^{n'} |V_{\safe (u_j)}|
$.
This together with (\ref{e:pint:1})
enables us to deduce (\ref{e:pint:3}).
Replacing $\ell_1$ with $\ell$ by Lemma \ref{l:Gl}.\ref{l:Gl:mon},
the orientation (\ref{e:pint:3}) together with Lemma \ref{l:Gl}.\ref{l:Gl:perm} now
 allows us to conclude 
$\PINT (G \seg y) \cspl{ab} \PINT (H \seg w_1)$.
\qed
\end{proof}

Let $G \in \mathcal{TG(F)}$ and
$\att_G (\rootnode_G) = v_1, \dots, v_k$.
We call $G$ a {\em basic} term graph if
$\lab_G (\rootnode_G) \in \DS$ and
$G \seg v_j \in \mathcal{TG(C)}$ for all
$j \in \{ 1, \dots, k \}$.
By definition,
$G \in \mathcal{TG}_{\nrm} (\FS)$
holds for any basic term graph $G \in \mathcal{TG(F)}$.

\begin{lemma}
\label{l:basic}
Let $\GS$ be an infinite set of constructor safe recursive unfolding graph rewrite rules over a signature $\FS$.
For any closed basic term graph
$G \in \mathcal{TG(F)}$,
if $G \rewast H$, then,
for any node $v \in V_H$ on a safe path from $\rootnode_H$,
$|\bigcup_{u \in \nrm (v)} V_{H \seg u}| \leq
 |\bigcup_{u \in \nrm (\rootnode_G)} V_{G \seg u}|$
holds.
\end{lemma}

\begin{proof}
Suppose $G \rewm{n} H$.
By induction on $n \geq 0$ we show that
for any node $v \in V_H$ on a safe path from $\rootnode_H$ with
$|\bigcup_{u \in \nrm (v)} V_{H \seg u}| \leq
 |\bigcup_{u \in \nrm (\rootnode_G)} V_{G \seg u}|$
holds.
In the base case $n=0$, $H = G$ and hence the assertion follows
 trivially.

For the induction step, suppose that
$G \rewm{n} H$ holds and that
$H \rew K$ is induced by a redex 
$(R, \varphi)$ in $H$ for a rewrite rule 
$R = (H', l, r) \in \GS$ and a homomorphism
$\varphi : H' \seg l \rightarrow H$.
Then 
$H, K \in \TGnrm$
holds by Lemma \ref{l:safepath}.\ref{l:safepath:2}.
Let a node $v \in V_K$ lie on a safe path from 
$\rootnode_K$.

{\sc Case.}
$\varphi (l) = \rootnode_H$:
In this case, 
$|\bigcup_{u \in \nrm (v)} V_{K \seg u}| \leq
 |\bigcup_{u \in \nrm (\rootnode_H)} V_{H \seg u}|$
follows from Corollary \ref{c:sruf}.
Since $\rootnode_H$ lies on the trivial safe path from $\rootnode_H$,
$|\bigcup_{u \in \nrm (\rootnode_H)} V_{H \seg u}| \leq
 |\bigcup_{u \in \nrm (\rootnode_G)} V_{G \seg u}|$
holds by induction hypothesis, and thus 
$|\bigcup_{u \in \nrm (v)} V_{K \seg u}| \leq
 |\bigcup_{u \in \nrm (\rootnode_G)} V_{G \seg u}|$
also holds.

{\sc Case.}
$\varphi (l) \neq \rootnode_H$:
Let
$r' \in H$ denote the node corresponding to $r \in H'$.
First consider the subcase $v \in V_{K \seg r'}$.
In this subcase, as in the previous case,
$|\bigcup_{u \in \nrm (v)} V_{K \seg u}| \leq
 |\bigcup_{u \in \nrm (\varphi (l))} V_{H \seg u}|$
follows from Corollary \ref{c:sruf}.
Since $\varphi (l)$ lies on a safe path from $\rootnode_H$ by Lemma
\ref{l:safepath}.\ref{l:safepath:1},
$|\bigcup_{u \in \nrm (\varphi (l))} V_{H \seg u}| \leq
 |\bigcup_{u \in \nrm (\rootnode_G)} V_{G \seg u}|$
holds by induction hypothesis, and thus 
$|\bigcup_{u \in \nrm (v)} V_{K \seg u}| \leq
 |\bigcup_{u \in \nrm (\rootnode_G)} V_{G \seg u}|$
holds.

Consider the subcase $v \not\in V_{K \seg r'}$.
As in the previous subcase, it holds that
\begin{equation}
\textstyle
 |\bigcup_{u \in \nrm (r')} V_{K \seg u}| \leq
 |\bigcup_{u \in \nrm (\varphi (l))} V_{H \seg u}|.
\label{e:basic:1}
\end{equation}
On the other side, since
$V_{K \seg v} \setminus V_{K \seg r'} =
 V_{H \seg v} \setminus V_{H \seg \varphi (l)}$,
it holds that
\begin{equation}
\textstyle
 |\bigcup_{u \in \nrm (v)} V_{K \seg u} \setminus
  \bigcup_{u \in \nrm (r')} V_{K \seg u}|
 \leq
 |\bigcup_{u \in \nrm (v)} V_{H \seg u} \setminus
  \bigcup_{u \in \nrm (\varphi (l))} V_{H \seg u}|.
\label{e:basic:2}
\end{equation}
Combining the inequalities (\ref{e:basic:1}) and (\ref{e:basic:2}), we
 reason as follows.
\begin{eqnarray*}
\textstyle
 |\bigcup_{u \in \nrm (v)} V_{K \seg u}|
 &=&
\textstyle
 |\bigcup_{u \in \nrm (v)} V_{K \seg u} \setminus
  \bigcup_{u \in \nrm (r')} V_{K \seg u}| +
  |\bigcup_{u \in \nrm (r')} V_{K \seg u}| \\
 &\leq&
\textstyle
 |\bigcup_{u \in \nrm (v)} V_{H \seg u} \setminus
  \bigcup_{u \in \nrm (\varphi (l))} V_{H \seg u}|
 +
 |\bigcup_{u \in \nrm (\varphi (l))} V_{H \seg u}|
\\
 &\leq&
\textstyle
 |\bigcup_{u \in \nrm (v)} V_{H \seg u}|
\\
 &\leq&
\textstyle
 |\bigcup_{u \in \nrm (\rootnode_G)} V_{G \seg u}|.
\end{eqnarray*}
The last inequality follows from induction hypothesis.
\qed
\end{proof}

\begin{theorem}
\label{t:context}
Let $\GS$ be an infinite set of constructor safe recursive unfolding graph rewrite rules over a signature $\FS$.
Suppose $\max \{ \arity (\mf) \mid \mf \in \FS \} \leq d$.
Then, for any closed basic term graph 
$G_0 \in \mathcal{TG(F)}$,
if $G_0 \rewast G$ and $G \rew H$, then 
$\PINT (G) \spl{\ell} \PINT (H)$
holds for
$\ell = 2 |\bigcup_{v \in \nrm (\rootnode_{G_0})} V_{G_0 \seg v}| + d$.
\end{theorem}

\begin{proof} 
Given a closed basic term graph
$G_0 \in \mathcal{TG(F)}$,
suppose that $G_0 \rewast G$ and that
$G \rew H$ is induced by a redex $(R, \varphi)$
in $G$ for a rule $R = (G', l, r)$ and a homomorphism
$\varphi: G' \seg l \rightarrow G$.
Then $G, H \in \mathcal{TG}_{\nrm} (\FS)$
holds by Lemma \ref{l:safepath}.\ref{l:safepath:2}.
Let
$\att_{G'} (l) = \sn{u_1, \dots, u_{k}}{u_{k+1}, \dots, u_{k+l}}$
and
$\att_{G} (\varphi (l)) = 
 \sn{v_1, \dots, v_{k}}{v_{k+1}, \dots, v_{k+l}}$.
Since
$|G' \seg r| \leq 2 |G' \seg u_1| + d$
holds by the definition of unfolding graph rewrite rules,
$|G' \seg r| \leq 2 |G' \seg u_1| + d \leq 2 |G \seg v_1| + d$
holds.
On the other hand, since $\varphi (l)$ lies on a safe path from
$\rootnode_G$ by Lemma \ref{l:safepath}.\ref{l:safepath:1},
$|\bigcup_{v \in \nrm (\varphi (l))} V_{G \seg v}| \leq 
 |\bigcup_{v \in \nrm (\rootnode_{G_0})} V_{G_0 \seg v}|$
holds by Lemma \ref{l:basic}.
Hence
$|G' \seg r| \leq
 2 |\bigcup_{v \in \nrm (\rootnode_{G_0})} V_{G_0 \seg v}| + d$
holds.
Now let 
$\ell = 2 |\bigcup_{v \in \nrm (\rootnode_{G_0})} V_{G_0 \seg v}| + d$.
In case $\varphi (l) = \rootnode_G$,
$\pint (G) \spl{\ell} \pint (H)$
 follows from Lemma \ref{l:pint} (and Lemma
\ref{l:Gl}.\ref{l:Gl:mon}). 
In case $\varphi (l) \neq \rootnode_G$,
$\pint (G) \spl{\ell} \pint (H)$ follows from
Lemma \ref{l:pint}, \ref{l:safepath}.\ref{l:safepath:1}
and \ref{l:Gl}.\ref{l:Gl:con}.
\qed
\end{proof}
 
\begin{corollary}
\label{c:main}
Let $\GS$ be an infinite set of constructor safe recursive
 unfolding graph rewrite rules over a signature $\FS$.
Then there exists a polynomial 
$p: \mathbb N \rightarrow \mathbb N$ such that, for any
 closed basic term graph $G \in \mathcal{TG(F)}$,
if
$G \rewm{m} H$ for some term graph $H \in \mathcal{TG(F)}$, then
$m \leq p (|\bigcup_{v \in \nrm (\rootnode_G)} V_{G \seg v}|)$
holds.
\end{corollary}

\begin{proof}
Given an infinite set $\GS$ of constructor safe recursive unfolding
 graph rewrite rules over a signature $\FS$,
let
$\max \{ \arity (\mf) \mid \mf \in \FS \} \leq d$.
In addition, given a closed basic term graph
$G \in \mathcal{TG(F)}$,
let
$\ell = 2 |\bigcup_{v \in \nrm (\rootnode_G)} V_{G \seg v}| + d$.
Suppose that $G \rewm{m} H$ holds for some term graph
$H \in \mathcal{TG(F)}$.
By Theorem \ref{t:context}, any
$\rew$ sequence starting with $G$ can be embedded into some
$\spl{\ell}$ reduction sequence starting with
$\pint (G)$,
and hence $m$ can be bounded by $\mG{\ell} (\PINT (G))$.
Write $\lst{f(s_1, \dots, s_k)}$ to denote
$\PINT (G)$.
Since $k \leq d \leq \ell$,
Lemma \ref{l:Gleq} implies
$\mG{\ell} (\PINT (G)) \leq
 d^{\rk (f)} \cdot (1+ \ell)^{\rk (f)} \cdot 
 \left( 1 + \sum_{j=1}^k \depth (s_j) \right)$.
Now let $p$ denote a polynomial such that
\[
 d^{\max \{ \rk (f) \mid f \in \FS \}} \cdot 
 (1 + 2 x +d)^{\max \{ \rk (f) \mid f \in \FS \}} \cdot (1 + d x) \leq
 p (x).
\]
For every $j \in \{ 1, \dots, k \}$, 
$s_j = \term (G \seg v)$ for some $v \in \nrm (\rootnode_G)$,
and hence
$\depth (s_j) \leq |G \seg v|$
holds for some $v \in \nrm (\rootnode_G)$.
Thus
$\sum_{j=1}^k \depth (s_j) \leq 
 d \cdot |\bigcup_{v \in \nrm (\rootnode_G)} V_{G \seg v}|$
holds.
This together with 
$\ell = 2 |\bigcup_{v \in \nrm (\rootnode_G)} V_{G \seg v}| + d$
allows us to conclude that
$m \leq \mG{\ell} (\PINT (G)) \leq 
p (|\bigcup_{v \in \nrm (\rootnode_G)} V_{G \seg v}|)$ holds.
\qed
\end{proof}
 
\begin{lemma}
\label{l:size}
Let $\GS$ be an infinite set of constructor safe recursive
 unfolding graph rewrite rules over a signature $\FS$.
For any closed basic term graph
$G \in \mathcal{TG(F)}$
and for any term graph
$H \in \mathcal{TG(F)}$, if
$G \rewm{n} H$, then
$|V_H \setminus \bigcup_{v \in \nrm (\rootnode_H)} V_{H \seg v}| \leq  
 n \cdot |\bigcup_{v \in \nrm (\rootnode_G)} V_{G \seg v}| +
 |V_G \setminus \bigcup_{v \in \nrm (\rootnode_G)} V_{G \seg v}|$
holds.  
\end{lemma}

\begin{proof} 
By induction on $n$.
In the base case $n=0$, $H = G$ and hence the assertion trivially holds.
For the induction step, suppose
$G \rewm{n} H$ and $H \rew K$.
Let us observe that for any safe recursive unfolding graph rewrite rule
$(G', l, r)$,
\begin{itemize}
\item $|V_{G' \seg r} \setminus 
        \bigcup_{v \in \nrm (\rootnode_{G'})} V_{G' \seg v}
       | \leq
       |G' \seg l|$, and
\item for any node $v \in V_{G' \seg r}$, if $\lab_{G'} (v)$ is
      undefined, then $v \in V_{G' \seg l}$.
\end{itemize} 
From the observation, it can be seen that
$|V_K \setminus \bigcup_{v \in \nrm (\rootnode_K)} V_{K \seg v}|
 \leq |H|$
holds.
Namely, the following inequality holds. 
\begin{equation}
\textstyle
 |V_K \setminus \bigcup_{v \in \nrm (\rootnode_K)} V_{K \seg v}| \leq  
 |\bigcup_{v \in \nrm (\rootnode_H)} V_{H \seg v}| +
 |V_H \setminus \bigcup_{v \in \nrm (\rootnode_H)} V_{H \seg v}|
\label{e:size:1}
\end{equation}
On the other hand, Lemma \ref{l:basic} yields
\begin{equation}
\textstyle
 |\bigcup_{v \in \nrm (\rootnode_H)} V_{H \seg v}| \leq
 |\bigcup_{v \in \nrm (\rootnode_G)} V_{G \seg v}|.
\label{e:size:2}
\end{equation}
Moreover, induction hypothesis yields 
\begin{equation}
\textstyle
 |V_H \setminus \bigcup_{v \in \nrm (\rootnode_H)} V_{H \seg v}| \leq  
 n \cdot |\bigcup_{v \in \nrm (\rootnode_G)} V_{G \seg v}| +
 |V_G \setminus \bigcup_{v \in \nrm (\rootnode_H)} V_{G \seg v}|.
\label{e:size:3}
\end{equation}
Combining the three inequalities (\ref{e:size:1}), (\ref{e:size:2}) and
(\ref{e:size:3}) allows us to conclude that
$|V_K \setminus \bigcup_{v \in \nrm (\rootnode_K)} V_{K \seg v}| \leq  
 (n+1) \cdot |\bigcup_{v \in \nrm (\rootnode_G)} V_{G \seg v}| +
 |V_G \setminus \bigcup_{v \in \nrm (\rootnode_G)} V_{G \seg v}|
$
holds.
\qed
\end{proof}

\begin{corollary}
\label{c:size}
Let $\GS$ be an infinite set of constructor safe recursive
 unfolding graph rewrite rules over a signature $\FS$.
Then there exists a polynomial 
$p: \mathbb N \rightarrow \mathbb N$ such that, for any
 closed basic term graph $G \in \mathcal{TG(F)}$
and for any term graph $H \in \mathcal{TG(F)}$, if
$G \rewast H$, then
$|H| \leq p (|\bigcup_{v \in \nrm (\rootnode_G)} V_{G \seg v}|) +
 |V_G \setminus \bigcup_{v \in \nrm (\rootnode_G)} V_{G \seg v}|$
holds.  
\end{corollary}

\begin{proof}
Given a closed basic term graph $G \in \mathcal{TG(F)}$,
suppose that $G \rewm{m} H$ holds for some term graph 
$H \in \mathcal{TG(F)}$,
Then, by Corollary \ref{c:main}, one can find a polynomial 
$q: \mathbb N \rightarrow \mathbb N$
such that
$m \leq q (|\bigcup_{v \in \nrm (\rootnode_G)} V_{G \seg v}|)$
holds.
Define a polynomial 
$p: \mathbb N \rightarrow \mathbb N$
by
$p(x) = (1+ q(x)) \cdot x$.
Then we conclude as follows.
\begin{eqnarray*}
 |H| &=&
\textstyle
 \left|\bigcup_{v \in \nrm (\rootnode_H)} V_{H \seg v}
 \right| +
 \left|V_H \setminus \bigcup_{v \in \nrm (\rootnode_H)} V_{H \seg v}
 \right| \\
 &\leq&
\textstyle
 \left|\bigcup_{v \in \nrm (\rootnode_G)} V_{G \seg v}
 \right| +
 \left|V_H \setminus \bigcup_{v \in \nrm (\rootnode_H)} V_{H \seg v}
 \right|
\qquad \quad (\text{by Lemma \ref{l:basic}}) \\
 &\leq&
\textstyle
 (1+m)
 \left| \bigcup_{v \in \nrm (\rootnode_G)} V_{G \seg v}
 \right| +
 \left| V_G \setminus \bigcup_{v \in \nrm (\rootnode_G)} V_{G \seg v}
 \right|  
\ (\text{by Lemma \ref{l:size}}) \\
 &\leq&
\textstyle
 \left( 1+ q \left( \left| \bigcup_{v \in \nrm (\rootnode_G)} V_{G \seg v} 
                    \right| 
             \right)
 \right)
 \cdot \left| \bigcup_{v \in \nrm (\rootnode_G)} V_{G \seg v} \right| \\
&&
\textstyle
\hspace{3.5cm}
 + \left| V_G \setminus \bigcup_{v \in \nrm (\rootnode_G)} V_{G \seg v}
   \right|
\quad (\text{by Corollary \ref{c:main}}) \\
 &=&
\textstyle
 p \left( \left| \bigcup_{v \in \nrm (\rootnode_G)} V_{G \seg v}
          \right|
   \right) +
 \left| V_G \setminus \bigcup_{v \in \nrm (\rootnode_G)} V_{G \seg v}
 \right|.  
\end{eqnarray*}
\qed
\end{proof}

As a consequence of Corollary \ref{c:main} and \ref{c:size},
for any set $\GS$ of constructor safe recursive unfolding graph rewrite
rules, we can find a polynomial
$p: \mathbb N \rightarrow \mathbb N$ such that, for any
 closed basic term graph $G \in \mathcal{TG(F)}$
and for any term graph $H \in \mathcal{TG(F)}$, if
$G \rewm{m} H$, then
$m  \leq p (|\bigcup_{v \in \nrm (\rootnode_G)} V_{G \seg v}|)$
and
$|H| \leq p (|\bigcup_{v \in \nrm (\rootnode_G)} V_{G \seg v}|) +
 |V_G \setminus \bigcup_{v \in \nrm (\rootnode_G)} V_{G \seg v}|$.  

\section{Interpreting general safe recursive functions}
\label{s:gsr}

Until the previous section, we have restricted every GRS to a set of
safe recursive graph rewrite rules.
In this section, to interpret all the general safe recursive functions,
expanding safe recursive unfolding graph rewrite rules,
we introduce
{\em safe recursive graph rewrite systems} by which every general safe
recursive function can be expressed.
We show that every safe recursive GRS can be interpreted into the
relation $\spl{\ell}$ by the predicative interpretation defined in the
previous section, sharpening the complexity result obtained in 
\cite{GRR2010}.
Let $\CS$ be a set of constructors and 
$m \mapsto c_m$ ($1 \leq m \leq |\CS|$) be an enumeration for $\CS$.
We assume that $\CS$ contains at least one constant.
We call a function 
$f: \mathcal{T(C)}^{k+l} \rightarrow \mathcal{T(C)}$
{\em general safe recursive} if, under a suitable argument separation
$f(\sn{x_1, \dots, x_k}{y_1, \dots, y_l})$,
$f$ can be defined from the initial functions by operating the schemata specified below.
\begin{itemize}
\item $O^{k, l}_j (\sn{x_1, \dots, x_k}{y_1, \dots, y_l}) = c_j$
      if $c_j$ is a constant.
\hfill {\bf (Constants)}
\item $C_j (; x_1, \dots, x_{\arity (c_j)}) = 
       c_j (x_1, \dots, x_{\arity (c_j)})$
      if $\arity (c_j) \neq 0$.
\hfill {\bf (Constructors)}
\item $I^{k, l}_j (\sn{x_1, \dots, x_k}{y_1, \dots, y_l}) = 
       \left\{
       \begin{array}{ll}
       x_j & \text{if } 1 \leq j \leq k, \\
       y_{j-k} & \text{if } k \leq j \leq k+l
       \end{array}
       \right.
      $
\hfill {\bf (Projections)}
\\
      $(1 \leq j \leq k+l)$.
\item $P_{i, 0} (\sn{}{c_i}) = c_i$ if $c_i$ is a constant. \\
      $P_{i, j} (\sn{}{c_i (x_1, \dots, x_{\arity (c_i)})
                      }
                )
       = x_j$ $(1 \leq j \leq \arity (c_i))$.
\hfill {\bf (Predecessors)}
\item $C(\sn{}%
            {c_{j} (x_1, \dots, x_{\arity (c_j)}), 
             y_1, \dots, y_{|\CS|}
            }) =
       y_{j}$.
\hfill {\bf (Conditional)}
\item $f(\sn{x_1, \dots, x_k}{y_1, \dots, y_l}) =
       h(\sn{x_{j_1}, \dots, x_{j_m}}%
            {g_1 (\sn{\vec x}{\vec y}), \dots,
             g_n (\sn{\vec x}{\vec y})
            }
        )$
\\
$(\{ j_1, \dots, j_m \} \subseteq \{ 1, \dots, k \})$,
\hfill {\bf (Safe composition)}\\
where $h$ has $m$ normal and $n$ safe argument positions.
\item $f(\sn{c_j (x_1, \dots, x_{\arity (c_j)}), \vec y}{\vec z}) =
       h_j (\sn{\vec x, \vec y}%
               {\vec z, f(\sn{x_1, \vec y}{\vec z}), \dots,
                        f(\sn{x_{\arity (c_j)}, \vec y}{\vec z})
               }
           )$
\\
$(j \in I)$
\hfill {\bf (General safe recursion)}\\
In case that $c_j$ is a constant, the schema of general safe recursion should be understood as
$f(\sn{c_j, \vec y}{\vec z}) = h_j (\sn{\vec y}{\vec z})$.
\end{itemize}

\begin{definition}[Safe recursive graph rewrite systems]
\label{d:srgrs}
\normalfont
We call a GRS $\GS$ over a signature $\FS$ {\em safe recursive} if $\GS$ consists of an infinite
 number of safe recursive unfolding graph rewrite rules over
 $\FS$ and of a finite
 number of graph rewrite rules $(G, l, r)$ fulfilling one of the
 following conditions
\ref{d:srgrs:1} and \ref{d:srgrs:2}.
\begin{enumerate}
\item $G \seg r = (G \seg l) \seg v$ for some node
      $v \in V_{G \seg l} \setminus \{ l \}$.
\label{d:srgrs:1}
\item 
\label{d:srgrs:2}
  \begin{enumerate}
  \item The set $V_G$ of vertices consists of 
        $2+k+l+n$ elements
        $u$, $v$, $x_1, \dots, x_{k+l}$,
        $w_1, \dots, w_n$.
  \item $l = u$ and $r = v$.
  \item $\{ \lab_G (u), \lab_G (v), \lab_G (w_1), \dots, \lab_G (w_n)
         \} \subseteq \FS$.
  \item $\lab_G (x_j)$ is undefined for all 
        $j \in \{ 1, \dots, k+l \}$.
  \item $\att_G (u) =
         \sn{x_1, \dots, x_k}{x_{k+1}, \dots, x_{k+l}}$.
  \item $\att_G (v) =
         \sn{x_{j_1}, \dots, x_{j_m}}{w_1, \dots, w_n}$
        for some
        $\{ j_1, \dots, j_m \} \subseteq \{ 1, \dots, k \}$.
  \item $\att_G (w_j) =
         \sn{x_1, \dots, x_k}{x_{k+l}, \dots, x_{k+l}}$
        for all $j \in \{ 1, \dots, n \}$.
  \end{enumerate}
\end{enumerate}
\end{definition}

\begin{lemma}
\label{l:gsr}
For every general safe recursive function $f$ over a finite set 
$\CS$ of constructors, there exists a constructor safe recursive GRS
 defining $f$.
\end{lemma}

\begin{proof}
By induction over the definition of $f$.
Since we always assume that constructors only have safe argument
 positions, we can identify the function $O^{0, 0}_j$ with the constant 
$c_j$ and the function $C_j$ with the constructor $c_j$.
In the base case, {\bf (Constants)} can be defined by a single graph rewrite
 rule in a special shape of \ref{d:srgrs:2} in Definition 
\ref{d:srgrs}, and each of
{\bf (Projections)}, {\bf (Predecessors)} and {\bf (Conditional)}
can be defined by a single graph rewrite rule in the form of 
\ref{d:srgrs:1} in Definition \ref{d:srgrs}.
The induction step splits into two cases.
In case that $f$ is defined by {\bf (Safe composition)},
$f$ is defined by a graph rewrite rule in the form of 
\ref{d:srgrs:2} in Definition \ref{d:srgrs} together with the
 constructor safe recursive GRSs obtained from induction hypothesis.
In case that $f$ is defined by {\bf (General safe recursion)},
$f$ is defined by an infinite set of constructor safe recursive unfolding graph
 rewrite rules together with the constructor safe recursive GRSs
 obtained from induction hypothesis.
\qed
\end{proof}

\begin{example}
\label{ex:GRS}
Let us discuss a safe recursive GRS expressing the TRS 
$\RS$ on page \pageref{ex:TRS}.
To obey the formal definition of safe recursive GRSs, instead of considering $\RS$ directly,
we consider the following TRS over the signature $\FS$
with 
$\CS = \{ \epsilon, \m{0}, \m{c}, \ms \}$ and
$\DS = \{ \mI{0,1}{1}, \mI{2,2}{3},  \mI{2,2}{4}, \mI{2,1}{3},
          \m{e}, \mh_0, \mh_1, \mg, \mf 
       \}$.

$\begin{array}{rclcrcl}
&&& \qquad &
\mI{2,l}{j} (\sn{x,y}{u_1, u_2}) & \rightarrow & u_{j-2} \quad
(l = 1, 2, \ j \in \{ 3, 2+l \}) \\
\mI{0,1}{1} (\sn{}{x}) & \rightarrow & x & &
\mh_0 (\sn{x, y}{\, u, v}) & \rightarrow & 
\m{c} (\sn{}{\mI{2,2}{3} (\sn{x, y}{u, v}), \mI{2,2}{4} (\sn{x, y}{u, v})}) \\
\mg (\sn{\epsilon}{z}) & \rightarrow & \mI{0,1}{1} (\sn{}{z}) & &
\mg (\sn{\m{c} (\sn{}{x, y})}{z}) & \rightarrow &
\mh_0 (\sn{x, y}{z, \mg (\sn{x}{z}), \mg (\sn{y}{z})}) \\
\me (\sn{x}{}) & \rightarrow & \epsilon & &
\mh_1 (\sn{x, y}{z}) & \rightarrow &
\mg (\sn{y}{\mI{2,1}{3} (\sn{x,y}{z})}) \\
\mf (\sn{\m{0}, y}{}) & \rightarrow & \me (\sn{y}{}) & &
\mf (\sn{\ms (\sn{}{x}), y}) & \rightarrow & 
\mh_1 (\sn{x, y}{\mf (\sn{x, y}{})})
\end{array}
$

\noindent
The rewrite rules defining
$\mI{0,1}{1}$, $\mI{2,2}{3}$, $\mI{2,2}{4}$ and $\mI{2,1}{3}$
can be expressed by graph rewrite rules in the shape of 
Case \ref{d:srgrs:1} in Definition \ref{d:srgrs}.
For example, the rule
$\mI{2,1}{3} (\sn{x,y}{z}) \rightarrow z$
is expressed by the rule $(1)$ below.
The defining rule for $\mh_0$ can be expressed by the graph rewrite rule
$(2)$ which is an instance of Case \ref{d:srgrs:2} in Definition
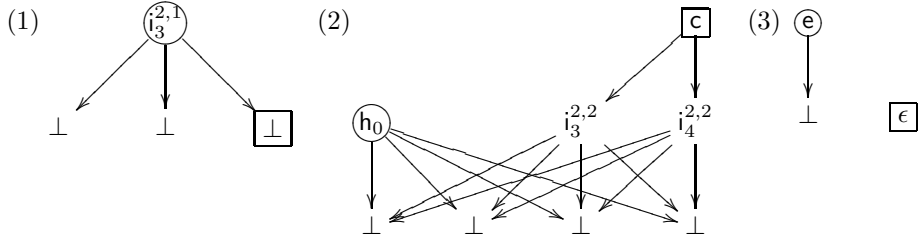
\ref{d:srgrs}.
\begin{equation*}
  (1)
  \xymatrix{& *+[o][F-]{\mI{2,1}{3}} \ar[dl] \ar[d] \ar[dr] & \\
            \bot & \bot & *+[F]{\bot}
           }
\quad
  (2)
  \xymatrix{ & & & *+[F]{\m{c}} \ar[dl] \ar[d] \\
            *+[o][F-]{\mh_0} \ar[d] \ar[dr] \ar[drr] \ar[drrr] & &
            \mI{2,2}{3} \ar[dll] \ar[dl] \ar[d] \ar[dr] & 
            \mI{2,2}{4} \ar[dlll] \ar[dll] \ar[dl] \ar[d] \\
            \bot & \bot & \bot & \bot
           }
\quad
  (3) \
  \xymatrix{*+[o][F-]{\m{e}} \ar[d] & \\
            \bot & *+[F]{\epsilon}
           }
\end{equation*}
Consider the rewrite rules
$\mg (\sn{\epsilon}{z}) \rightarrow \mI{0,1}{1} (\sn{}{z})$
and
$\mg (\sn{\m{c} (\sn{}{x, y})}{z}) \rightarrow
 \mh_0 (x, y; z$, $\mg (\sn{x}{z}), \mg (\sn{y}{z}))$.
Let $\Sigma_{\mg} = \{ \epsilon, \m{c} \}$ and 
$\Theta_{\mg} = \{ \mI{0,1}{1}, \mh_0 \}$ be two signatures with the bijection
$\epsilon \mapsto \mI{0,1}{1}$ and $\m{c} \mapsto \mh_0$.
Define an argument separation as indicted in the rules above.
Then, the rewrite rules defining $\mg$ can be expressed by the set 
$\GS_{\mg}$ of all the safe recursive
 unfolding graph rewrite rules defining $\mg$ over
$\Sigma_{\mg} \cup \Theta_{\mg} \cup \{ \mg \}$.
The rule
$\m{e} (\sn{x}{}) \rightarrow \epsilon$
can be expressed by the graph rewrite rule $(3)$ above which is a
 special case of Case \ref{d:srgrs:2} in Definition
\ref{d:srgrs}.
The defining rule for $\mh_1$ can be expressed by a graph rewrite rule
 as $(2)$.
Finally, consider the rewrite rules 
$\mf (\sn{\m{0}, y}{}) \rightarrow \m{e} (\sn{y}{})$
and
$\mf (\sn{\ms (\sn{}{x}), y}) \rightarrow  
 \mh_1 (\sn{x, y}{\mf (\sn{x, y}{})})
$.
Let 
$\Sigma_{\mf} = \{ \m{0}, \ms \}$ and 
$\Theta_{\mf} = \{ \m{e}, \mh_1 \}$ be two signatures with the bijection
$\m{0} \mapsto \m{e}$ and $\ms \mapsto \mh_1$.
Define an argument separation as indicted accordingly.
Then, the rewrite rules defining $\mf$ can be expressed by 
the infinite set of safe recursive unfolding graph rewrite rules
 defining $\mf$ over 
$\Sigma_{\mg} \cup \Theta_{\mg} \cup \{ \mf \}$.
Now, define a GRS $\GS$ by
$\GS = \GS_{\mg} \cup \GS_{\mf} \cup \GS_0$,
where $\GS_0$ is the finite set of graph rewrite rules defining
$\mI{0,1}{1}$, $\mI{2,2}{3}$,  $\mI{2,2}{4}$, $\mI{2,1}{3}$,
$\m{e}$, $\mh_0$ and $\mh_1$ as pictured above.
Clearly $\GS$ is a constructor safe recursive GRS, and the TRS can
 be expressed by $\GS$.
\end{example}

\begin{lemma}
\label{l:pintgsr}
Let $\GS$ be a constructor safe recursive GRS over a signature $\FS$.
Suppose that $G \rew H$ is induced by a redex $(R, \varphi)$ in a closed term graph 
$G \in \mathcal{TG}_{\nrm} (\FS)$
for a rule $R = (G', l, r) \in \GS$ and a homomorphism
$\varphi : G' \seg l \rightarrow G$.
Let $r' \in V_H$ the node corresponding to $r \in V_{G'}$.
Then, for the interpretations $\pint$ defined for $G$ and $H$,
$\PINT (G \seg \varphi (l)) \spl{\ell} 
 \PINT (H \seg r')$
holds for
$\ell =
\max (\{ |G' \seg r| \} \cup \{ \arity (f) \mid f \in \FS \})$.
\end{lemma}

\begin{proof}
By Lemma \ref{l:pint}, it suffices to check the case that the rule
$R \in \GS$ is in the form either \ref{d:srgrs:1} or
\ref{d:srgrs:2} in Definition \ref{d:srgrs}.
We mention that Lemma \ref{l:safepath} still holds for a constructor
 safe recursive GRS $\GS$. 

{\sc Case  \ref{d:srgrs:1}}.
$H \seg r' = (G \seg \varphi (l)) \seg \varphi (v) =
 G \seg \varphi (v)$ for some node
$v \in V_{G' \seg l} \setminus \{ l \}$.
If $G \seg \varphi (v) \in \mathcal{TG(C)}$, then
$H \seg r' \in \mathcal{TG(C)}$, and hence
$\pint (G \seg \varphi (l)) \cspl{ab} \lst{} =
 \pint (H \seg r')$.
Suppose $G \seg \varphi (v) \not\in \mathcal{TG(C)}$.
Since $G \in \TGnrm$, the node $\varphi (v)$ lies on a safe path from
$\rootnode_G$.
Moreover, as well as $\varphi (l)$, the node $r'$ lies on a safe path
 from $\rootnode_H$.
From these observations, the equality
$H \seg r' = G \seg \varphi (v)$
and Lemma \ref{l:Gl}.\ref{l:Gl:con},
one can show that
$\pint (G \seg \varphi (l)) \spl{\ell} \pint (H \seg r')$
holds by structural induction over 
$G \seg \varphi (l)$.

{\sc Case \ref{d:srgrs:2}}.
Let $V_{G'}$ consists of $2+k+l+n$ elements
$u$, $v$, $x_1, \dots, x_{k+l}$, $w_1, \dots, w_n$
as specified in Case \ref{d:srgrs:2} in Definition \ref{d:srgrs}.
To make the presentation simpler, let us identify the nodes
$u$, $x_1, \dots, x_{k+l} \in V_{G' \seg l}$
with the nodes in $V_{G}$ corresponding by the homomorphism
$\varphi$
and the nodes
$v$, $w_1, \dots, w_n$
$x_{k+1}, \dots, x_{k+l} \in V_{G' \seg r}$
with the corresponding nodes in $V_H$, 
e.g., 
$u = \varphi (l)$ and $v = r'$.
We write $f$ to denote $\lab_G (u)$, $h$ to denote 
$\lab_{H} (v)$ and write $g_j$ to denote
$\lab_H (w_j)$ for each $j \in \{ 1, \dots, n \}$.
Then, by the interpretations $\pint$ defined for $G$ and $H$, the term graphs $G \seg u$ and $H \seg v$ are respectively transformed into
 the following sequences of terms.
\begin{eqnarray*}
\PINT (G \seg u) &=&
\lst{\fn (\term (G \seg x_1), \dots, \term (G \seg x_k))
    }
\con 
\PINT (G \seg y_{1}) \con \cdots \con
\PINT (G \seg y_{l'}),
\\
\PINT (H \seg v) &=&
\lst{\hn (\term (H \seg x_{j_1}), \dots, \term (H \seg x_{j_m})}
\con 
\PINT (H \seg w_{1}) \con \cdots \con
\PINT (H \seg w_{n}),
\\
\PINT (H \seg w_1) &=&
\lst{\fsn{(g_1)} (\term (H \seg x_1), \dots, \term (H \seg x_k))
    }
\con 
\PINT (H \seg y_{1}) \con \cdots \con
\PINT (H \seg y_{l'}),
\\
\PINT (H \seg w_j) &=&
\lst{\fsn{(g_j)} (\term (H \seg x_1), \dots, \term (H \seg x_k))
    }
\hspace{2cm} \hfill (2 \leq j \leq n),
\end{eqnarray*}
where 
$\{ j_1, \dots, j_m \} \subseteq \{ 1, \dots, k \}$,
$\{ y_1, \dots, y_{l'} \} =
 \{ x \in \{ x_{k+1}, \dots, x_{k+l} \} \mid
          (u, x) \in E_{\pj (G) \seg u}
 \}
 =
 \{ x \in \{ x_{k+1}, \dots$, $x_{k+l} \} \mid
    (w_1, x) \in E_{\pj (H) \seg w_1}
 \}$,
and
$y_1 < \cdots < y_{l'}$.
Define a precedence $>$ as
 $\fn > \hn$, $\fsn{(g_1)}, \dots, \fsn{(g_n)}$.
Write $t_j$ to denote
$\term (G \seg x_j)$ for each $j \in \{ 1, \dots, k \}$.
Then
$\term (H \seg x_j) = t_j$ holds for all
$j \in \{ 1, \dots, k \}$.
Since $\fn > \hn$ and 
$\fn (t_1, \dots, t_k) \rhd t_{j_m}$
for all $i \in \{ 1, \dots, m \}$,
the following orientation holds.
\begin{equation}
\label{e:pintgsr:1}
\fn (t_1, \dots, t_k)
\cspl{st}
\hn (t_{j_1}, \dots, t_{j_m})
\end{equation}
Similarly, since $\fn > \fsn{(g_j)}$ for all 
$j \in \{ 1, \dots, n \}$ and
$\fn (t_1, \dots, t_k) \rhd t_{j}$
for all $j \in \{ 1, \dots, k \}$,
the following orientations also hold.
\begin{equation}
\label{e:pintgsr:2}
\fn (t_1, \dots, t_k)
\cspl{st}
\fsn{(g_j)} (t_{1}, \dots, t_{k})
\quad (j=1, \dots, n)
\end{equation}
Since $1+n \leq |G' \seg r| \leq \ell$,
the orientations (\ref{e:pintgsr:1}) and (\ref{e:pintgsr:2}) imply
\begin{equation*}
\fn (t_1, \dots, t_k)
\cspl{sb}
\lst{\hn (t_{j_1}, \dots, t_{j_m}) \
     \fsn{(g_1)} (t_{1}, \dots, t_{k}) \cdots
     \fsn{(g_n)} (t_{1}, \dots, t_{k})
    }.
\end{equation*}
This together with Lemma  \ref{l:Gl}.\ref{l:Gl:perm} allows us to
conclude 
$\pint (G \seg u) \cspl{ab} \pint (H \seg v)$.
\qed
\end{proof}

One would observe that Lemma \ref{l:basic} also holds for a constructor
safe recursive GRS $\GS$.
Thus, by a slight modification of the proof of Theorem \ref{t:context},
using Lemma \ref{l:pintgsr} instead of Lemma \ref{l:pint},
one can deduce the following theorem.

\begin{theorem}
\label{t:contextgsr}
Let $\GS$ be a constructor safe recursive GRS over a signature $\FS$ and
$\GS_0$ be the maximal subset of $\GS$ that does not contain any
 unfolding graph rewrite rule.
Suppose
$\max (\{ \arity (f) \mid f \in \FS \} \cup
       \{ |G \seg r| \mid \exists l \ (G, l, r) \in \GS_0 \}
      ) \leq d$.
For any closed basic term graph
$G_0 \in \mathcal{TG(F)}$,
if $G_0 \rewast G$ and $G \rew H$, then 
$\PINT (G) \spl{\ell} \PINT (H)$
holds for
$\ell = 2 |\bigcup_{v \in \nrm (\rootnode_{G_0})} V_{G_0 \seg v}| + d$.
\end{theorem}
 
\begin{corollary}
\label{c:gsr:1}
For every general safe recursive function $f$ over constructors
$\CS$, there exist a constructor safe recursive GRS $\GS$ over a signature
$\FS \supseteq \CS$ defining $f$ and a polynomial 
$p: \mathbb N \rightarrow \mathbb N$ such that, for any
 closed basic term graph $G \in \mathcal{TG(F)}$,
if
$G \rewm{m} H$ for some term graph $H \in \mathcal{TG(F)}$, then
$m \leq p (|\bigcup_{v \in \nrm (\rootnode_G)} V_{G \seg v}|)$
holds.
\end{corollary}

In addition, Lemma \ref{l:size} can be modified as follows.

\begin{lemma}
\label{l:gsr:size}
Let $\GS$ be a constructor safe recursive GRS over a signature $\FS$,
$\GS_0$ be the maximal subset of $\GS$ that does not contain any
 unfolding graph rewrite rule, and
$\max \{ |G \seg r| \mid \exists l \ (G, l, r) \in \GS_0 \}
 \leq d$.
For any closed basic term graph $G \in \mathcal{TG(F)}$
and for any term graph $H \in \mathcal{TG(F)}$,
if
$G \rewm{n} H$, then
$|V_H \setminus \bigcup_{v \in \nrm (\rootnode_H)} V_{H \seg v}| \leq  
 n \cdot \left( |\bigcup_{v \in \nrm (\rootnode_G)} V_{G \seg v}| +d \right) +
 |V_G \setminus \bigcup_{v \in \nrm (\rootnode_G)} V_{G \seg v}|$
holds.  
\end{lemma}

In order to show that a witnessing GRS $\GS$ for Theorem \ref{t:GRR10} is
polynomially bounded, it is shown that there exists a
polynomial $p: \mathbb N \rightarrow \mathbb N$ such that
$\max \{ m, |H| \} \leq p (|G|)$ holds whenever $G \irewm{m} H$ holds
\cite[Propositon 1]{GRR2010}.
As a consequence of Corollary \ref{c:gsr:1} and Lemma \ref{l:gsr:size},
it can be sharpened as follows.

\begin{corollary}
\label{c:gsr:2}
For every general safe recursive function $f$ over constructors $\CS$, there exist a
constructor safe recursive GRS $\GS$ over a signature 
$\FS \supseteq \CS$ defining $f$ and a polynomial
$p: \mathbb N \rightarrow \mathbb N$ such that,
for any closed basic term graph
$G \in \mathcal{TG(F)}$
and for any term graph $H \in \mathcal{TG(F)}$,
if $G \rewm{m} H$, then the following two conditions hold.
\begin{enumerate}
\item $m \leq p \left( |\bigcup_{v \in \nrm (\rootnode_G)} V_{G \seg v}| \right)$.
\item $|H| \leq p \left( |\bigcup_{v \in \nrm (\rootnode_G)} V_{G \seg v}| \right) +
       |V_G \setminus \bigcup_{v \in \nrm (\rootnode_G)} V_{G \seg v_j}|$.
\end{enumerate}
\end{corollary}

In contrast to Theorem \ref{t:GRR10}, the upper bound
$p \left( |\bigcup_{v \in \nrm (\rootnode_G)} V_{G \seg v}| \right)$
for $m$ depends only on the size
$|\bigcup_{v \in \nrm (\rootnode_G)} V_{G \seg v}|$
(of the union) of the subgraphs connected to
the normal argument positions.
Moreover, innermost rewriting is not assumed as long as rewriting starts
with a (closed) basic term graph.

\begin{remark}
The schema (\ref{e:gsr}) is formulated based on safe recursion (on
 notation) following \cite{BC92} whereas
the schema of general ramified recurrence formulated in \cite{GRR2010}
 is based on ramified recurrence following \cite{Leivant95}.
Due to the difference between safe recursion and ramified recurrence, the
 definition of general safe recursive functions on page \pageref{d:srgrs}
 is slightly different from the original definition of
{\em tiered recursive} functions in \cite{GRR2010}.
Notably, the schema (\textbf{Safe composition}) is a weaker form of the
 original one in \cite{BC92}, which was introduced in \cite{HW99}.
It is not clear whether there is a precise correspondence between
 general safe recursive functions in the current formulation and tiered
 recursive functions.
However, it is known that the polytime functions (over binary words) can
 be covered with the weak form of safe composition, cf.
\cite[Lemma 12]{AEMspop}, which means that the restriction of the
 general safe recursive functions to unary constructors still
covers all the polytime functions.
\end{remark}

\section{Precedence termination with argument separation}
\label{s:ptas}

In this section, generalising the definition of safe recursive
GRSs defined in the previous section, we propose
{\em precedence termination with argument separation},
which is a restriction of the standard precedence termination in the
sense of \cite{MidOZ96}.
The restrictive precedence termination together with suitable
assumptions yields a new criterion for
the polynomial runtime complexity of infinite GRSs and for the polynomial
size of normal forms in infinite GRSs 
(Corollary \ref{c:spt:2}).

\begin{definition}
\label{d:ptas}
\normalfont
Let $>$ be a precedence over a signature $\FS$ and
$G, H \in \mathcal{TG(F)}$ be two term graphs.
Then the relation $G \spt H$ holds if 
$\lab_G (\rootnode_G) > \lab_H (v)$
for any $v \in V_H$ whenever $\lab_H (v)$ is defined, and additionally
one of the following two cases holds.
\begin{enumerate}
\item $G \seg v \eqspt H$ for some successor node
      $v$ of $\rootnode_G$.
\item $\lab_H (\rootnode_H)$ is defined, i.e.
      $\lab_G (\rootnode_G) > \lab_H (\rootnode_H)$,
      \begin{itemize}
      \item for each $v \in \nrm (\rootnode_H)$,
            $H \seg v$ is a sub-term graph of 
            $G \seg u$ for some $u \in \nrm (\rootnode_G)$, and
      \item $G \spt H \seg v$ for each $v \in \safe (\rootnode_H)$.
      \end{itemize}
\end{enumerate} 
We say that a GRS $\GS$ over a signature $\FS$ is
{\em precedence-terminating with argument separation}
if for some separation of argument positions and for some precedence
$>$ on $\FS$, the following two conditions are fulfilled.
\begin{enumerate}
\item For each rule $(G, l, r) \in \GS$,
      $\lab_G (v)$ is undefined for any node
      $v \in \safe (l)$.
\item $G \seg l \spt G \seg r$ for each rule 
      $(G, l, r) \in \GS$ for the relation 
      $\spt$ induced by the precedence $>$. 
\end{enumerate}
\end{definition}

We note that, for any graph rewrite rule $(G, l, r)$,
{\em variable nodes are maximally shared} in the term graph
$G \seg r$, which may not be assumed in a different formulation of graph
rewrite systems.
In the rest of this section, for every GRS $\GS$ we always assume that for each
rewrite rule $(G, l, r) \in \GS$ and for any node 
$v \in V_{G \seg r}$, if $\lab_G (v)$ is undefined, then
$v \in V_{G \seg l}$.
We only consider GRSs over finite signatures.
Hence, for any (infinite) constructor GRS $\GS$ over a signature
$\FS = \CS \cup \DS$,
the defined symbols $\DS$ can be partitioned into two sets
$\Dinf$ and $\Dfin$
so that every symbol $f \in \Dinf$ is defined by infinite rules
whereas every symbol $f \in \Dfin$ is defined by finite rules.

Recall that Lemma \ref{l:Gl}.\ref{l:Gl:perm} was employed to show Lemma 
\ref{l:pint}.
To show a lemma corresponding to Lemma \ref{l:pint}, Lemma
\ref{l:Gl}.\ref{l:Gl:perm} is slightly generalised.
 
\begin{lemma}
\label{l:Gl:permptas}
If $s \cspl{sb} b$,
$b = b_1 \con \cdots \con b_{k}$ and
$b_j \neq \lst{}$ for each $j \in \{ 1, \dots, k \}$, then
$\lst{s} \con c \cspl{ab}
 b_1 \con c_1 \con  \cdots \con b_k \con c_k$
holds for any sequence
$c = c_1 \con \cdots \con c_k$.
\end{lemma}

\begin{lemma}
\label{l:pintptas}
{\normalfont (Cf. Lemma \ref{l:pint})}
Let $\GS$ be a constructor GRS over a signature $\FS$ that is
 precedence-terminating with argument separation.
Suppose that $G \rew H$ is induced by a redex $(R, \varphi)$ in a closed term graph 
$G \in \mathcal{TG}_{\nrm} (\FS)$
for a rule $R = (G', l, r) \in \GS$ and a homomorphism
$\varphi : G' \seg l \rightarrow G$.
Let $r' \in V_H$ the node corresponding to $r \in V_{G'}$.
Then, for the interpretations $\pint$ defined for $G$ and $H$,
$\PINT (G \seg \varphi (l)) \spl{\ell} 
 \PINT (H \seg r')$
holds for
$\ell =
\max (\{ |G' \seg r| \} \cup \{ \arity (f) \mid f \in \FS \})$.
\end{lemma}

\begin{lemma}
\label{l:basicptas}
{\normalfont (Cf. Lemma \ref{l:basic})}
Let $\GS$ be a constructor GRS over a signature $\FS$ that is
 precedence-terminating with argument separation.
For any closed basic term graph
$G \in \mathcal{TG(F)}$,
if $G \rewast H$, then,
for any node $v \in V_H$ on a safe path from $\rootnode_H$,
$|\bigcup_{u \in \nrm (v)} V_{H \seg u}| \leq
 |\bigcup_{u \in \nrm (\rootnode_G)} V_{G \seg u}|$
holds.
\end{lemma}

\begin{theorem}
\label{t:contextptas}
{\normalfont (Cf. Theorem \ref{t:context})}
Let $\GS$ be a constructor GRS over a signature $\FS$ that is
 precedence-terminating with argument separation and
$\GS_0 = \{ (G, l, r) \in \GS \mid 
             \lab_G (l) \in \Dfin
         \}$.
Suppose the following two conditions.
\begin{enumerate}
\item $\max \left( \{ \arity (f) \mid f \in \FS \} \cup
             \{ |G \seg r| \mid \exists l \ (G, l, r) \in \GS_0 \}
            \right) \leq d$.
\label{t:contextptas:1}
\item $|V_{G \seg r} \setminus \bigcup_{v \in \nrm (r)} V_{G \seg v}| \leq
       |G \seg l|$
      for any rule $(G, l, r) \in \GS \setminus \GS_0$.
\label{t:contextptas:2}
\end{enumerate}
Then, for any closed basic term graph
$G_0 \in \mathcal{TG(F)}$,
if $G_0 \rewast G$ and $G \rew H$, then 
$\PINT (G) \spl{\ell} \PINT (H)$
holds for
$\ell = 2 |\bigcup_{v \in \nrm (\rootnode_{G_0})} V_{G_0 \seg v}| + d$.
\end{theorem}
 
\begin{corollary}
\label{c:spt:1}
{\normalfont (Cf. Corollary \ref{c:main})}
Let $\GS$ be a constructor GRS over a signature $\FS$ that is
 precedence-terminating with argument separation.
Then there exists a polynomial 
$p: \mathbb N \rightarrow \mathbb N$ such that, for any
 closed basic term graph $G \in \mathcal{TG(F)}$,
if
$G \rewm{m} H$ holds for some term graph
$H \in \mathcal{TG(F)}$, then
$m \leq p (|\bigcup_{v \in \nrm (\rootnode_G)} V_{G \seg v}|)$
holds.
\end{corollary}

\begin{lemma}
\label{l:spt:size}
{\normalfont (Cf. Lemma \ref{l:size})}
Let $\GS$ be a constructor GRS over a signature $\FS$ that is
 precedence-terminating with argument separation and
$\GS_0 = \{ (G, l, r) \in \GS \mid 
             \lab_G (l) \in \Dfin
         \}$.
Suppose the conditions \ref{t:contextptas:1} and  \ref{t:contextptas:2}
 in Theorem \ref{t:contextptas} hold.
For any closed basic term graph $G \in \mathcal{TG(F)}$
and for any term graph $H \in \mathcal{TG(F)}$,
if
$G \rewm{n} H$, then
$|V_H \setminus \bigcup_{v \in \nrm (\rootnode_H)} V_{H \seg v}| \leq  
 n \cdot \left( |\bigcup_{v \in \nrm (\rootnode_G)} V_{G \seg v}| +d \right) +
 |V_G \setminus \bigcup_{v \in \nrm (\rootnode_G)} V_{G \seg v}|$
holds.  
\end{lemma}

\begin{corollary}
\label{c:spt:2}
{\normalfont (Cf. Corollary \ref{c:size})}
Suppose that $\GS$ is an infinite constructor GRS over a signature $\FS$
 precedence-terminating with argument separation that enjoys
the condition \ref{t:contextptas:2} in Theorem \ref{t:contextptas}.
Then, for any closed basic term graph
$G \in \mathcal{TG(F)}$
and for any term graph $H \in \mathcal{TG(F)}$,
if $G \rewm{m} H$, then the following two conditions hold.
\begin{enumerate}
\item $m \leq p \left( |\bigcup_{v \in \nrm (\rootnode_G)} V_{G \seg v}| \right)$.
\item $|H| \leq p \left( |\bigcup_{v \in \nrm (\rootnode_G)} V_{G \seg v}| \right) +
       |V_G \setminus \bigcup_{v \in \nrm (\rootnode_G)} V_{G \seg v}|$.
\end{enumerate}
\end{corollary}

\begin{example}
Let us consider the GRS $\GS$ defined in Example \ref{ex:GRS}. 
For the signature $\FS$ with the partition
$\CS = \{ \epsilon, \m{0}, \m{c}, \ms \}$ and
$\DS = \{ \mI{0,1}{1}, \mI{2,2}{3},  \mI{2,2}{4}, \mI{2,1}{3},
          \m{e}, \mh_0, \mh_1, \mg, \mf 
       \}$,
define a precedence $>$ by

$\begin{array}{rclcrclcrclcrclcrcl}
 \mf &>& \me & \quad & \me &>& \epsilon & \quad & \mh_1 &>& \mg &
 \quad & \mg &>& \mI{0,1}{1} & \quad & \mh_0 &>& \m{c} \\
 \mf &>& \mh_1 & \quad & & & & \quad & \mh_1 &>& \mI{2,1}{3} &
 \quad & \mg &>& \mh_0 & \quad & \mh_0 &>& \mI{2,2}{j} \ (j=3,4)
 \end{array}
$

\noindent
It is easy to see that $>$ is well-founded.
It is routine to check that for each rule $(G, l, r) \in \GS$,
\begin{itemize}
\item $\lab_G (v)$ is undefined for any node $v \in \safe (l)$, and
\item $\lab_G (l) > \lab_G (v)$ for any node $v \in V_{G \seg r}$
      whenever $\lab_G (v)$ is defined.
\end{itemize}
Let $\spt$ the relation induced by the precedence $>$ as defined in
 Definition \ref{d:ptas}.
Recall that $\GS$ is defined by 
$\GS = \GS_{\mg} \cup \GS_{\mf} \cup \GS_0$, 
where $\GS_{\mg}$ and $\GS_{\mf}$ are infinite sets of safe
recursive unfolding graph rewrite rules respectively defining
$\mg$ and $\mf$, and $\GS_0$ is a finite set of the rewrite rules defining the
 other function symbols.
Therefore, the set $\DS$ is partitioned into
$\Dfin = \{ \mI{0,1}{1}, \mI{2,2}{3},  \mI{2,2}{4}, \mI{2,1}{3},
          \m{e}, \mh_0, \mh_1
       \}$
and
$\Dinf = \{ \mg, \mf\}$.
Note that $\GS_0$ does not contain any unfolding graph rewrite rule.
It follows from the definition of safe recursive unfolding graph rewrite
 rules that
$G \seg l \spt G \seg r$
for each $(G, l, r) \in \GS_{\mg} \cup \GS_{\mf}$
(See also Corollary \ref{c:sruf}).
Consider a rewrite rule $(G, l, r) \in \GS_0$.
It is obvious that 
$G \seg l \spt G \seg r$
holds if $(G, l, r,)$ is an instance of Case \ref{d:srgrs:1} in
 Definition \ref{d:srgrs}.
Suppose that $V_G$ consists of $2+k+l+n$ elements
$u$, $v$, $x_1, \dots, x_{k+l}$, $w_1, \dots, w_n$
as specified in Case
\ref{d:srgrs:2} in Definition \ref{d:srgrs}. 
Let
$v \in V_{G \seg r} = 
 \{ v, x_1, \dots, x_{k+l}, w_1, \dots, w_n \}$.
Consider the case that $\lab_G (v)$ is undefined, i.e.,
$v \in \{ x_1, \dots, x_{k+l} \}$.
In this case, $v$ is a successor node of $l$.
Namely $G \seg v = G \seg u$ for some successor node $u$ of $l$, and
 hence $G \seg l \spt G \seg v$ holds.
Assume that $\lab_G (v) \in \FS$.
Then $v \in \{ v, w_1, \dots, w_n \}$.
Since
$\att_G (w_j) = \sn{x_1, \dots, x_k}{x_{k+1}, \dots, x_{k+l}}$
for every $j \in \{ 1, \dots, n \}$,
it holds that
\begin{equation}
G \seg l \spt G \seg w_j. 
\quad (j \in \{ 1, \dots, n \})
\label{e:ex:GRS}
\end{equation}
Since 
 $\att_G (v) = \sn{x_{j_1}, \dots, x_{j_m}}{w_1, \dots, w_n}$
for some
$\{ j_1, \dots, j_m \} \subseteq \{ 1, \dots, k \}$,
it follows from (\ref{e:ex:GRS}) that
$G \seg l \spt G \seg v$ holds.
To apply Corollary \ref{c:spt:2}, we also have to check that $\GS$
 enjoys the condition
\ref{t:contextptas:2} in Theorem \ref{t:contextptas}. 
As a consequence of the partition 
$\DS = \Dfin \cup \Dinf$ 
observed above, the finite subset $\GS_0$ coincides with the set
$\{ (G, l, r) \in \GS \mid \lab_G (l) \in \Dfin \}$.
In other words,
$\GS \setminus \GS_0 = \GS_{\mg} \cup \GS_{\mf}$.
Thus, it suffices to show that
$|V_{G \seg r} \setminus \bigcup_{v \in \nrm (r)} V_{G \seg v}|
 \leq |G \seg l|$
for any rule $(G, l, r) \in \GS_{\mg} \cup \GS_{\mf}$,
but it follows from the definition of safe recursive unfolding graph
 rewrite rules.
Therefore, the runtime complexity and the sizes of normal forms in $\GS$
 can be polynomially bounded as in Corollary \ref{c:spt:2}. 
\end{example}

\section{Conclusion}

In this paper we introduced a termination order over sequences of terms
together with an interpretation of term graphs into
sequences of terms.
Unfolding graph rewrite rules which express the equation of
(\ref{e:gsr}) can be successfully embedded into the termination order by
the interpretation, sharpening the result obtained in \cite{GRR2010}
about the runtime complexity of those unfolding graph rewrite rules.
The introduction of the termination order is strongly motivated by
former works \cite{AM05,AM08,AEM11,AEM12} and also based on an
observation that every unfolding graph rewrite rule is 
precedence terminating in the sense defined in \cite{MidOZ96}.
Generalising the definition of unfolding graph
rewrite rules for general safe recursion, we proposed a restrictive
notion of the standard precedence termination,
precedence termination with argument separation.
The restrictive precedence termination together with suitable assumptions
yields a new criterion for
the polynomial runtime complexity of infinite GRSs and for the polynomial
size of normal forms in infinite GRSs.

\subsection*{Acknowledgments}

The author thanks Kazushige Terui for drawing his attention to the work
\cite{GRR2010}, which was the initial point of the current research.


\end{document}